\documentclass[journal]{IEEEtran}
\title{Multi-Channel Factor Analysis: Identifiability and Asymptotics}
\author{
Gray~Stanton,~\IEEEmembership{Student Member,~IEEE,
}
\and
David~Ram{\'i}rez,~\IEEEmembership{Senior Member,~IEEE,
}
\and
Ignacio~Santamaria,~\IEEEmembership{Senior Member,~IEEE,
}
\and
Louis~Scharf,~\IEEEmembership{Life Fellow,~IEEE,
}
and
Haonan~Wang%
\thanks{
This paper is the journal version of \cite{Stanton2023}.  

Gray Stanton and Haonan Wang are with the Department of Statistics, Colorado State University, Fort Collins, CO 80523 USA (e-mail: gstanton@colostate.edu; wanghn@stat.colostate.edu).

David Ram{\'i}rez is with the Department of Signal Theory and Communications, Universidad Carlos III de Madrid, Madrid 28903, Spain, and also with Gregorio Mara\~{n}\'{o}n Health Research Institute, Madrid 28007, Spain (e-mail:
david.ramirez@uc3m.es).

Ignacio Santamaria is with the Department of Communications Engineering, Universidad de Cantabria, 39005 Santander, Spain (e-mail:
i.santamaria@unican.es).

Louis Scharf is with the Department of Mathematics, Colorado State University, Fort Collins, CO 80523 USA (e-mail: louis.scharf@colostate.edu).

© 2024 IEEE.  Personal use of this material is permitted.  Permission from IEEE must be obtained for all other uses, in any current or future media, including reprinting/republishing this material for advertising or promotional purposes, creating new collective works, for resale or redistribution to servers or lists, or reuse of any copyrighted component of this work in other works.

This paper has supplementary downloadable material available at https://ieeexplore.ieee.org/document/10596994, provided by the author. The material includes additional proofs. This material is 2 pages in size.
}}

\usepackage{amsmath}
\usepackage[utf8]{inputenc}
\usepackage{orcidlink}
\usepackage[english]{babel}
\usepackage{amssymb}
\usepackage{graphicx}
\usepackage{amsthm}
\usepackage{microtype}
\usepackage{commath}
\usepackage{subfig}
\usepackage{setspace}
\usepackage{nicematrix}

\IEEEoverridecommandlockouts
\usepackage{cite}
\usepackage{amsmath,amssymb,amsfonts}
\usepackage{algorithmic}
\usepackage{graphicx}
\usepackage{textcomp}
\usepackage{xcolor}
\usepackage{multirow}

\renewcommand{\vec}[1]{\mathbf{#1}}

\usepackage{graphics}
\usepackage[ruled,vlined]{algorithm2e}

\newcommand{\Var}{\mathop{\textnormal{Var}}}

\renewcommand{\vec}[1]{\mathbf{#1}}

\newcommand{\tr}{\mathop{\mathrm{tr}}}

\newcommand{\argmin}{\mathop{\mathrm{argmin}}}

\newcommand{\bvarphi}{\boldsymbol{\varphi}}

\newcommand{\bSigma}{\boldsymbol{\Sigma}}

\newcommand{\brho}{\boldsymbol{\rho}}

\newcommand{\bPhi}{\boldsymbol{\Phi}}
\newcommand{\tran}{\mathsf{T}}
\newcommand{\rank}{\mathop{\mathrm{rank}}}
\newcommand{\blkdiag}{\mathop{\mathrm{blkdiag}}}
\newcommand{\kernel}{\mathop{\mathrm{Ker}}}
\newcommand{\imagesp}{\mathop{\mathrm{Im}}}
\newcommand{\veta}{\boldsymbol{\eta}}

\newcommand{\Sym}{\mathop{\mathrm{Sym}}}

\newtheorem{theorem}{Theorem}
\newtheorem{lemma}{Lemma}
\newtheorem{proposition}{Proposition}

\newtheorem{condition}{Condition}

\begin{document}
	\maketitle

\begin{abstract}\label{sec:abstract}
Recent work by Ram{\'i}rez et al. \cite{ramirez2020} has introduced Multi-Channel Factor Analysis (MFA) as an extension of factor analysis to multi-channel data that allows for latent factors common to all channels as well as factors specific to each channel. This paper validates the MFA covariance model and analyzes the statistical properties of the MFA estimators. In particular, a thorough investigation of model identifiability under varying latent factor structures is conducted, and sufficient conditions for generic global identifiability of MFA are obtained. The development of these identifiability conditions enables asymptotic analysis of estimators obtained by maximizing a Gaussian likelihood, which are shown to be consistent and asymptotically normal even under misspecification of the latent factor distribution.

\end{abstract}

\section{Introduction}
\label{sec:intro}

Factor analysis (FA) is a statistical technique for modeling second-order structure within a collection of measurements. The method explains an observed vector through an unobserved systemic part (which is typically of scientific or engineering interest) and an unobserved noise part. Classical or \emph{single-channel} factor analysis was originally developed within the field of psychometrics by Spearman \cite{spearman1904} as a method to identify a small number of unobserved random \emph{factors} which explain the between-individual variation in psychometric scores. In signal processing, FA and its extensions \cite{Klami2015,Sardarabadi2018,Ramirez2022} are employed in the uncalibrated setting where the noise variance is anisotropic and unknown \cite{Pesavento2001,Ramirez2011,Koutrouvelis2019}.  

A recent extension to factor analysis by Ram{\'i}rez et al.  \cite{ramirez2020} is of central interest to this paper. These authors developed \emph{multi-channel} factor analysis (MFA), which enables joint factor analysis of observations collected from \emph{multiple channels}. Many problems and associated methods possess a natural channel structure\cite{Cochran1990,Cochran1995,Ramirez2010,Santamaria2017}, such as grouping multi-sensor data by sensor modality \cite{Correa2010,Bhinge2017}. In MFA, some factors, termed \emph{common} factors, may influence all channels. In addition to common factors, in MFA each channel may also possess \emph{distinct} factors influencing that channel alone. MFA then decomposes the vector of observations into latent vectors that can be described as a signal that influences all channels, within-channel interference, and idiosyncratic noise. This decomposition is of value, as detecting a weak signal that presents across all channels in the presence of channel-specific interference and noise is a goal in domains such as passive radar \cite{Bandiera2007,Hack2014}, speech recognition\cite{Zwyssig2015,wen2022}, and astronomy \cite{Leshem2001,antman2020}.

Previous work on MFA has provided an estimation procedure for the parameters of the model based on likelihood-maximization under normality assumptions. However, for the output of the procedure to be meaningful, it is crucial that MFA be guaranteed to be identifiable using only what is known \emph{a priori}, namely the channel sizes and the dimensions of the signal and interference subspaces. For single-channel factor analysis, practitioners correctly assume identifiability whenever the number of common factors is much smaller than the number of observations \cite{shapiro1985}. However, for multi-channel factor analysis, the channel sizes and desired number of common and distinct factors may vary widely, and so the question of identifiability becomes more challenging and less intuitive. In \cite{ramirez2020}, the question of identifiability is recognized and some necessary conditions on the maximum number of common and distinct factors are discussed. This paper extends that discussion by carefully examining the two main sources of non-identifiability of MFA, namely isolation of the idiosyncratic noise variances and separation of signal and interference covariances.

The purpose of this paper is to give identifiability guarantees requiring only the specification of the channel sizes and signal and interference dimensionality. The asymptotic properties of the MFA estimators are then derived, which provides the previously-missing theoretical underpinnings for interpretation of the MFA parameter estimates. This parallels the advancement of single-channel FA as a statistical method as reviewed in Section \ref{sec:identifiability}. 
The main contributions of this paper to MFA are
\begin{enumerate}
    \item Necessary and sufficient conditions for separation of signal and interference covariances.
    \item Sufficient conditions on the number of common and distinct factors for generic global identifiability.
    \item Proof of the asymptotic consistency and normality of estimators derived from Algorithm 1 in \cite{ramirez2020}.
\end{enumerate}

The sufficient conditions for generic global identifiability of the MFA covariance model ensure that, for reasonable numbers of common and {distinct} factors, the decomposition of the observation covariance into common, {distinct}, and idiosyncratic parts will be unique for almost all population covariance matrices.  With this identifiability result, parameter estimates obtained by maximizing a Gaussian likelihood are shown to be consistent and asymptotically normal, even in the case where the true distribution of the latent vectors is non-normal. 

\subsection{Notation}\label{sec:notation}
Matrices and vectors are denoted with bold-faced symbols, and scalars are denoted with light-face symbols. A real matrix of size $n \times m$ is written as $\vec{D} \in \mathbb{R}^{n\times m}$, and a column vector of length $n$ is written as  $\vec{d} \in \mathbb{R}^{n}$. The zero matrix of dimension $m\times n$ is  $\vec{0}_{m,n}$ and the $n\times n$ identity matrix is $\vec{I}_n$. A zero vector of dimension $n$ is written as $\vec{0}_n$.  When clear from context, the subscripts may be dropped. The standard basis for $\mathbb{R}^n$ will be written as $\{\vec{e}_{1}, \dots, \vec{e}_{n}\}$. Matrix and vector transposes are written as $\vec{D}^{\tran}$ and $ \vec{d}^{\tran}$ respectively. The determinant of $\vec{D}$ is written as $\det \vec{D}$, and the trace is written as $\tr\vec{D}$. The $(i,j)$th entry of a matrix $\vec{D}$ is $[\vec{D}]_{ij}$, and similarly for the $i$th entry of a column vector. The matrix obtained by concatenating the columns of $\vec{B}$ to the right of the columns of $\vec{A}$ is written as $[\vec{A}\ \vec{B} ]$. For row and column index sets $\alpha \subset \{1,\dots, n\}$ and $\beta \subset \{1,\dots, m\}$, the submatrix $\vec{D}[\alpha, \beta]$ contains the entries $[\vec{D}]_{ij}$ with $(i,j) \in \alpha \times \beta$.  The non-negative part of a scalar expression $a \in  \mathbb{R}$ is $(a)_{+} \equiv \max\{a, 0\}$. {The normal distribution with mean $\vec{m}$ and variance $\vec{V}$ is $\mathcal{N}(\vec{m}, \vec{V})$}.

The operator $\mathrm{Diag}^{-1}$ applied to a matrix yields the vector containing the diagonal entries. The block-diagonal operator $\mathrm{blkdiag}$ applied to a list of matrices yields the block-diagonal matrix with the listed blocks. The $\mathrm{vec}$ operator vectorizes a matrix by stacking the columns vertically, while $\mathrm{vech}$, applicable to square matrices, vectorizes by extracting only the lower-triangular entries. Denote the spaces of $n \times n$ symmetric, symmetric positive semidefinite, orthogonal, and diagonal matrices as $\mathrm{Sym}(n)$, $\mathrm{PSD}(n)$, $\mathrm{O}(n)$ and $\mathrm{Diag}(n)$ respectively, with $\mathrm{Diag}_{\geq 0}(n)$ being $\mathrm{PSD}(n) \cap \mathrm{Diag}(n)$. For matrix $\vec{A}$ with submatrix $\vec{A}_{11}$, the generalized Schur complement of $\vec{A}_{11}$ is $\vec{A} \setminus \vec{A}_{11}$. {For symmetric matrices $\vec{V}, \vec{W}$ of the same size $\vec{V} \succeq \vec{W}$ indicates that $\vec{V} - \vec{W}$ is positive semi-definite}.
For vector subspaces $\mathcal{A}$ and $\mathcal{B}$, the subspace intersection is $\mathcal{A} \cap \mathcal{B}$ and the subspace sum and direct sum are respectively denoted by $\mathcal{A} + \mathcal{B}$ and $\mathcal{A} \oplus \mathcal{B}$. For a linear map $\vec{T}$, the image and kernel subspaces are $\imagesp(\vec{T})$ and $\kernel(\vec{T})$.

\section{Model}
\label{sec:model}

    \subsection{Description}
    An archetypal data collection scheme for which MFA is applicable consists of multiple sensors or observation units, each of which collects a vector of measurements. Often these sensors are homogeneous (such as when all sensors measure voltage), but MFA is also applicable to a heterogeneous collection of sensors. The input from an individual sensor then composes an individual \emph{channel} of observation for some shared signal which is measured by multiple sensors.  This channel structure is set by the design of the sensor array, and is known in advance of data collection. The channels are numbered by $c = 1,\dots, C$ with $n_c$ scalar measurements in channel $c$.

	For channel $c$, denote the vector of measurements within that channel as $\vec{x}_{c}$. The generative model for $\vec{x}_c$ is
	\begin{equation}
	\label{eq:stationarymfa_channel}
	\vec{x}_{c} = \vec{A}_{c}\vec{f} + \vec{B}_{c}\vec{g}_{c} + \vec{u}_{c},
	\end{equation}
    where $\vec{A}_c\vec{f}$ is the signal in channel $c$, $\vec{B}_{c}\vec{g}_{c}$ is the channel-$c$ interference {that lives within a low-dimensional subspace}, and $\vec{u}_c$ is the measurement noise. The matrices $\vec{A}_c \in \mathbb{R}^{n_c \times r_0}$ and $\vec{B}_c \in \mathbb{R}^{n_c \times r_c}$ are the common and {distinct} factor loadings for channel $c$.
	The number of common factors $r_0 \leq n$ and {distinct} factors $r_1,\dots, r_C$ with $r_c \leq n_c$ determine the flexibility of the model, as the common factor $\vec{f}$ is in $\mathbb{R}^{r_0}$ and the {distinct} factor for channel $c$, $\vec{g}_{c}$, is in $\mathbb{R}^{r_c}$. The remaining portion of each measurement in channel $c$ which is not a result of the influence of the latent factors is contributed by $\vec{u}_{c} \in \mathbb{R}^{n_c}$.
 
    {
    The above data collection scheme and related model \eqref{eq:stationarymfa_channel} is appropriate for several signal processing problems. In passive radar \cite{Hark2014}, the observations are collected from two multi-sensor arrays which make up the \emph{reference} and \emph{surveillance} channels. The common signal $\vec{f}$ affects both channels as $\vec{A}_1\vec{f}$ and $\vec{A}_2\vec{f}$ when a target is reflected by an opportunistic illuminator. As the multi-sensor arrays are spatially separated, the interferences can be modeled as the uncorrelated terms $\vec{B}_1\vec{g}_1$ and $\vec{B}_2\vec{g}_2$. Finally, the measurements are contaminated by uncorrelated noises whose variances are unknown in the absence of an accurate calibration. Another possible application of \eqref{eq:stationarymfa_channel} is cooperative relaying in Time-Division Multiple Access (TDMA) systems \cite{Dohler2010}, where multiple relays each transmit a common signal to a multi-antenna access point in sequential time slots $c=1,\dots, C$. The common signal presents in slot $c$ as $\vec{A}_c\vec{f}$ and is subject to interference $\vec{B}_c\vec{g}_c$ and measurement noise $\vec{u}_c$.}

	The all-channel observation vector is obtained by stacking the channels as $\vec{x} \equiv [\vec{x}_{1}^{\tran} \cdots \vec{x}_{c}^{\tran}]^{\tran}$. The first-order model for the all-channel observations is
	\begin{equation}
	\label{eq:stationarymfa}
	\vec{x} = \vec{A} \vec{f} + \vec{B} \vec{g} + \vec{u},\\
    \end{equation}
    where $\vec{A}$ and $\vec{B}$ are the all-channel loadings,
    \begin{equation}\label{eq:mfaparams}
	\vec{A} \equiv 
	\big [
	\vec{A}_1^{\tran}\ \cdots 
	\vec{A}_C^{\tran}
	\big ]^{\tran},
	\quad
	\vec{B} \equiv  \mathrm{blkdiag}(\vec{B}_1, \dots, \vec{B}_C),
    \end{equation}
	with $\vec{g} \equiv  [\vec{g}_{1}^{\tran} \cdots \vec{g}_{C}^{\tran}]^{\tran}$ and $\vec{u} \equiv  [\vec{u}_{1}^{\tran} \cdots \vec{u}_{C}^{\tran}]^{\tran}$. 
    For clarity of notation, let $\vec{n} \equiv [n_1, \dots, n_C]$, and $\vec{r} \equiv [r_0, r_1, \dots, r_C]$. The total number of observations is $n \equiv \sum_{c=1}^{C} n_c$ and the total number of {distinct} factors is $r \equiv \sum_{c=1}^{C} r_c$. Denote the cumulative sum of the number of observations and {distinct} factors as $n_{<c} \equiv \sum_{k=1}^{c-1} n_k$ and $r_{<c} \equiv \sum_{k=1}^{c-1} r_k$, respectively. For $c=1$, $r_{<1}$ and $n_{<1}$ are set to $0$. Similarly, define $n_{>c} \equiv n - n_c - n_{<c}$ and $r_{>c} \equiv r - r_c - r_{<c}$.  {Table \ref{tab:notation} summarizes the commonly used notation}.
	
	\subsection{Covariance Specification}\label{sec:specification}
	In \eqref{eq:stationarymfa}, the factors $\vec{f}, \vec{g}$ and the errors $\vec{u}$ are unobserved random quantities while the factor loadings $\vec{A}, \vec{B}$ are fixed unknown parameters. The latent factors are assumed to satisfy
    \vspace*{-0.3em}
	\begin{align*}\label{eq:factormoments}
	E[\vec{f}] &= \vec{0}_{r_0}, & E[\vec{f}\vec{f}^{\tran}] & \equiv \vec{R}_{\vec{f}\vec{f}}, \\ \mkern-10mu
	E[\vec{g}_c] &= \vec{0}_{r_c}, & E[\vec{g}_{c}\vec{g}_{c}^{\tran}] &\equiv \vec{R}_{\vec{g}_c\vec{g}_c}.
	\end{align*}
	Factors of different types are required to be uncorrelated,
     \vspace*{-0.3em}
	\begin{align*}
	E[\vec{f}\vec{g}^{\tran}] &= \vec{0}_{r_0,r}, & E[\vec{g}_{c} \vec{g}^{\tran}_{c'} ] &= \vec{0}_{r_c,r_{c'}} \   c \neq c'.
	\end{align*}
	The idiosyncratic errors $\vec{u}$ are assumed to satisfy
     \vspace*{-0.3em}
	\begin{align*}
	E[\vec{u}] &= \vec{0}_{n}, \mkern-10mu& E[\vec{u}\vec{g}^{\tran}] &= \vec{0}_{n,r}, \mkern-10mu \\ E[\vec{u}\vec{f}^{\tran}] &= \vec{0}_{n,r_0}, \mkern-10mu&  E[\vec{u}\vec{u}^{\tran}] &= \bPhi, 
	\end{align*}
	for some covariance matrix $\bPhi \in \mathrm{Diag}_{\geq 0}(n)$.
    Under the above specification, $\vec{x}$ has mean zero with covariance matrix
    \begin{equation}\label{eq:obscov}
    \begin{aligned}
    \vec{R}_{\vec{x}\vec{x}} &\equiv \vec{A}\vec{R}_{\vec{f}\vec{f}}\vec{A}^{\tran} + \vec{B}\vec{R}_{\vec{g}\vec{g}}\vec{B}^{\tran} + \bPhi,\\
    &= \vec{R}_{\vec{s}\vec{s}} + \vec{R}_{\vec{i}\vec{i}} + \bPhi,
    \end{aligned}
    \end{equation}
    where $\vec{R}_{\vec{g}\vec{g}} = \blkdiag(\vec{R}_{\vec{g}_1\vec{g}_1}, \dots, \vec{R}_{\vec{g}_C\vec{g}_C})$ and the \emph{signal} and \emph{interference} covariances are $\vec{R}_{\vec{s}\vec{s}} \equiv \vec{A}\vec{R}_{\vec{f}\vec{f}}\vec{A}^{\tran}$ and  $\vec{R}_{\vec{i}\vec{i}} \equiv \vec{B}\vec{R}_{\vec{g}\vec{g}}\vec{B}^{\tran}$, respectively. The set $\mathcal{R}(\vec{n}, \vec{r}) \subset \mathrm{PSD}(n)$ contains all observation covariance matrices realizable by \eqref{eq:obscov}.

    \subsection{Parameterization of MFA Covariance Models}\label{sec:parameterization}
    For given channel sizes $\vec{n}$ and factor numbers $\vec{r}$, the generative model \eqref{eq:stationarymfa} for the all-channel observation $\vec{x}$ under MFA determines a set of covariance matrices $\mathcal{R}(\vec{n}, \vec{r}) \subset \mathrm{PSD}(n)$ by \eqref{eq:obscov}. The set $\mathcal{R}(\vec{n}, \vec{r})$ can be parameterized in three ways, namely by the triple of structured components $(\vec{R}_{\vec{s}\vec{s}}, \vec{R}_{\vec{i}\vec{i}}, \bPhi)$, by the loading matrices $\vec{A}, \vec{B}$ and noise variances $\bPhi$ whose structures are shown in Figure \ref{fig:mfa_ident_params}, or by a vector $\veta$ which captures the degrees of freedom in the $(\vec{A}, \vec{B}, \bPhi)$ parameterization.

    \subsubsection{Parameterization by $(\vec{R}_{\vec{s}\vec{s}}, \vec{R}_{\vec{i}\vec{i}}, \bPhi)$} In MFA, \eqref{eq:obscov} shows that the observation covariance $\vec{R}_{\vec{x}\vec{x}}$ is the sum of a low-rank matrix $\vec{R}_{\vec{s}\vec{s}}$, a channel-structured block-diagonal matrix $\vec{R}_{\vec{i}\vec{i}}$ with low-rank blocks, and a non-negative diagonal matrix $\bPhi$. Any triple ($\vec{R}_{\vec{s}\vec{s}}, \vec{R}_{\vec{i}\vec{i}}, \bPhi)$ of appropriately structured $n\times n$ matrices determines an element of $\mathcal{R}(\vec{n}, \vec{r})$ by the second line of \eqref{eq:obscov}. That is, if $\vec{R}_{\vec{s}\vec{s}} \in \mathrm{PSD}(n)$ has rank at most $r_0$, $\vec{R}_{\vec{i}\vec{i}} \in \mathrm{PSD}(n)$ is block-diagonal whose $c$th block is $n_c \times n_c$ with rank at most $r_c$, and $\bPhi$ is in $\mathrm{Diag}_{\geq 0}(n)$, then 
    \begin{equation}
    	\vec{R}_{\vec{x}\vec{x}}(\vec{R}_{\vec{s}\vec{s}}, \vec{R}_{\vec{i}\vec{i}}, \bPhi) \equiv \vec{R}_{\vec{s}\vec{s}} + \vec{R}_{\vec{i}\vec{i}} + \bPhi
    \end{equation}
    is in $\mathcal{R}(\vec{n}, \vec{r})$. This can be seen by taking $\vec{R}_{\vec{f}\vec{f}}$ and $\vec{R}_{\vec{g}\vec{g}}$ to be identity matrices and obtaining $\vec{A}$ and $\vec{B}$ from the Cholesky factors of $\vec{R}_{\vec{s}\vec{s}}$ and $\vec{R}_{\vec{i}\vec{i}}$ respectively. 
    
    Recovering $(\vec{R}_{\vec{s}\vec{s}}, \vec{R}_{\vec{i}\vec{i}}, \bPhi)$ from an estimate of  $\vec{R}_{\vec{x}\vec{x}}$ is the central goal of MFA, as decomposing $\vec{R}_{\vec{x}\vec{x}}$ into the three summands will separate  $\vec{R}_{\vec{s}\vec{s}}$, which controls the cross-channel covariance, from $\vec{R}_{\vec{i}\vec{i}}$, which modifies the within-channel covariance. Both covariance-controlling components are then isolated from the idiosyncratic noise variance for individual inputs. As the summands are separately interpretable and are identifiable from $\vec{R}_{\vec{x}\vec{x}}$, as will be shown in Section \ref{sec:identifiability}, the parameterization of $\mathcal{R}(\vec{n}, \vec{r})$ in terms of $(\vec{R}_{\vec{s}\vec{s}}, \vec{R}_{\vec{i}\vec{i}}, \bPhi)$ forms the basis for interpreting the results of MFA.

    \subsubsection{Parameterization by $(\vec{A}, \vec{B}, \bPhi)$}\label{sec:loadingparam}However, the rank constraints on $\vec{R}_{\vec{s}\vec{s}}$ and $\vec{R}_{\vec{i}\vec{i}}$ are inconvenient, as the set of such matrices is not a vector space. It is typical in factor analysis to  parameterize in terms of the loading matrices $\vec{A}$ and $\vec{B}$, so that the rank constraints are automatically satisfied. This increases the complexity of the parameterization map (as it is quadratic rather than linear), but simplifies the domain. 
    
    The first line of \eqref{eq:obscov} parameterizes $\mathcal{R}(\vec{n}, \vec{r})$ in terms of $(\vec{A}, \vec{B}, \vec{R}_{\vec{f}\vec{f}}, \vec{R}_{\vec{g}\vec{g}}, \bPhi)$.  However, without further information about either the loading matrices $\vec{A}, \vec{B}$ or the factor variances $\vec{R}_{\vec{f}\vec{f}}, \vec{R}_{\vec{g}\vec{g}}$, it is clear that the pairs $(\vec{A}, \vec{R}_{\vec{f}\vec{f}})$ and $(\vec{B}, \vec{R}_{\vec{g}\vec{g}})$ are non-identifiable from knowledge of $\vec{x}$ alone. As the factors are unobserved, any change of basis on the factor spaces taking $(\vec{A}, \vec{f})$ to $(\vec{A}\vec{T}_0, \vec{T}^{-1}_0\vec{f})$ and $(\vec{B}_c, \vec{g}_c)$ to $(\vec{B}_c\vec{T}_c, \vec{T}^{-1}_c\vec{g}_c)$ leaves the observations unchanged. In the exploratory case where no information beyond the channel structure and the factor space dimensionality is assumed, the invariance of the observations to linear transformations of the factor space is most easily resolved by imposing that the factors be uncorrelated and unit-scale, $\vec{R}_{\vec{f}\vec{f}} = \vec{I}_{r_0}$ and $\vec{R}_{\vec{g}_c\vec{g}_c} = \vec{I}_{r_c}$ for all $c=1,\dots, C$. 

    Under this assumption, $\mathcal{R}(\vec{n}, \vec{r})$ can be parameterized as
    \begin{equation}
    \begin{aligned}
    \label{eq:covariancemodel}
	\vec{R}_{\vec{x}\vec{x}}(\vec{A}, \vec{B}, \bPhi) &\equiv \vec{A}\vec{A}^\tran + \vec{B}\vec{B}^{\tran} + \bPhi.
    \end{aligned}
	\end{equation}
    The set of common factor loadings $\vec{A}$ is $\mathbb{A} \equiv \mathbb{R}^{n\times r_0}$, while the set of {distinct} factor loadings is $\mathbb{B} \subset \mathbb{R}^{n \times r}$ containing those $\vec{B} \in \mathbb{R}^{n \times r}$ which are block diagonal with $c$th block $\vec{B}_{c} \in \mathbb{R}^{n_c \times r_c}$. The domain of $\vec{R}_{\vec{x}\vec{x}}(\vec{A}, \vec{B}, \bPhi)$ is $\mathbb{A} \times \mathbb{B} \times \mathrm{Diag}_{\geq 0}(n)$.

    \begin{figure}
        \centering
        \includegraphics{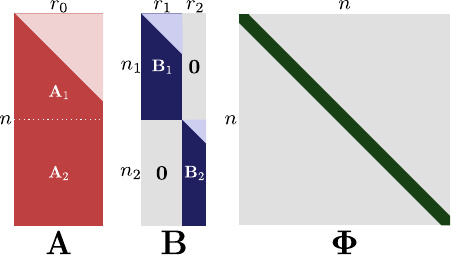}
        \caption{{Depiction of MFA covariance parameters $(\vec{A}, \vec{B}, \bPhi)$ for two channels. Triangles indicate constraints of $\mathbb{A}_{L}$ and $\mathbb{B}_{L}$.}}
        \label{fig:mfa_ident_params}
    \end{figure}

    \begin{table}
    \caption{{Table of commonly used notation and descriptions.}}
    \resizebox{\linewidth}{!}{
    \begin{tabular}{|c|c| c| c|}
    \hline
    Quantity & Description & Quantity & Description\\
    \hline
    $C$ & \# of channels & $n_c$ & Chan-$c$ size \\
    $r_0$ & Common fac. num.  & $r_c$ & Chan-$c$ distinct fac. num. \\
    $n$ & Total chan. size & $r$ & Total distinct fac. num.\\ 
    $\vec{n} \in \mathbb{N}^{C}$ & Vector of chan. sizes & $\vec{r}\in \mathbb{N}^{C{+}1}$ & Vector of fac. numbers\\
    $\vec{A} \in \mathbb{R}^{n \times r_0}$ & Common fac. loadings & $\vec{B}\in \mathbb{R}^{n\times r}$ & Distinct fac. loadings \\
    $\mathbb{A} \subset \mathbb{R}^{n\times r_0}$ & Set of $\vec{A}$s & $\mathbb{B} \subset{\mathbb{R}^{n\times r}}$ & Set of $\vec{B}$s \\
    $\mathbb{A}_{L}, \mathbb{B}_L$ & LT top block subspaces & $\mathbb{A}_L^*, \mathbb{B}_L^*$ & LT with positive diag.\\
    $\bPhi \in \mathbb{R}^{n\times n}$ & Diag. noise cov. & $\veta \in \mathbb{R}^{L}$ & $(\vec{A}, \vec{B},\bPhi)$ free params \\ 
    $\vec{R}_{\vec{s}\vec{s}} \in \mathbb{R}^{n\times n}$ & Rank-$r_0$ cov. & $\vec{R}_{\vec{i}\vec{i}} \in \mathbb{R}^{n\times n}$&  Blkdiag interf. cov. \\
    $\vec{R}_{\vec{x}\vec{x}}$ & MFA observation cov. & $\mathcal{R}(\vec{n}, \vec{r})$ & Set of $\vec{R}_{\vec{x}\vec{x}}$s\\ 
    \hline
    \end{tabular}
    }
    
    \label{tab:notation}
    \end{table}

    \subsubsection{Parameterization by $\veta$}
    The above parameterization $\vec{R}_{\vec{x}\vec{x}}(\vec{A}, \vec{B},\bPhi)$ by the loading matrices introduces a \emph{rotation invariance}, as $\vec{R}_{\vec{x}\vec{x}}(\vec{A}\vec{Q}_f, \vec{B}\vec{Q}_g, \bPhi)$ equals  $\vec{R}_{\vec{x}\vec{x}}(\vec{A}, \vec{B}, \bPhi)$ for any orthogonal $\vec{Q}_f$ and $\vec{Q}_g$, where $\vec{Q}_g = \blkdiag(\vec{Q}_{g,1}, \dots, \vec{Q}_{g, C})$ with $\vec{Q}_{g,c} \in \mathbb{R}^{r_c\times r_c},\ c=1,\dots, C$. For purposes of estimation and asymptotic analysis, it is desirable to eliminate this invariance by adding artificial restrictions to $\vec{A}$ and $\vec{B}$, in such a way that the realizable products $\vec{A}\vec{A}^{\tran}$ and $\vec{B}\vec{B}^{\tran}$ are not restricted. Analogous restrictions which remove rotation invariance in single-channel factor analysis are well-known, and typically involve orthogonality of loading matrix columns or imposition of structural zeros \cite{anderson1956}\cite{joreskog1969}. 
    
    
   Here, appropriate restrictions are imposed in the same fashion as in \cite{ramirez2020}. Consider loading matrices $\vec{A}$, $\vec{B}_1, \dots, \vec{B}_C$ which have lower-triangular (LT) top blocks, $\vec{A}_1$ and $\vec{B}_{1,1},\dots, \vec{B}_{C,1}$, of sizes $r_0 \times r_0$ and $r_c \times r_c, c=1,\dots, C$ respectively. The remaining rows are unconstrained, and the remaining submatrices are written as $\vec{A}_2$ and $\vec{B}_{1,2}, \dots, \vec{B}_{C,2}$.  Define $\mathbb{A}_{L} \subset \mathbb{A}$ and $\mathbb{B}_{L} \subset \mathbb{B}$ as the subspaces of loading matrices which satisfy their respective restrictions. Further, distinguish $\mathbb{A}^*_{L}$  as the subset where, for each $j=1,\dots,r_0$, the main diagonal element $[\vec{A}_1]_{jj}$ is either positive or the $j$th column of $\vec{A}_1$ is zero. The set $\mathbb{B}_L^*$ is defined similarly. 
    
    The non-redundant degrees of freedom in $(\vec{A}, \vec{B}, \bPhi)$ for $\vec{A} \in \mathbb{A}_L$ and $\vec{B} \in \mathbb{B}_L$ compose the vector $\veta \in \mathbb{R}^L$ as
     \begin{equation}
	\label{eq:vectorization}
	\begin{aligned}
	\veta = \big[&\mathrm{vech}(\vec{A}_1)^{\tran}\ \mathrm{vec}(\vec{A}_2)^{\tran}\ \mathrm{vech}(\vec{B}_{1,1})^{\tran}\\
	 &\mathrm{vec}(\vec{B}_{1, 2})^{\tran}\ \dots\ \mathrm{vec}(\vec{B}_{C,2})^{\tran}\ \mathrm{Diag}^{-1}(\bPhi)^{\tran} \big]^{\tran},
	\end{aligned}
	\end{equation}
    where the dimension $L$ is
	\begin{equation}
	\label{eq:unknown_count}
	L = nr_0 - \textstyle\frac{1}{2}r_0(r_0-1) + \textstyle\sum_{c=1}^C \big[n_cr_c - \frac{1}{2} r_c(r_c-1)\big] + n.
	\end{equation}
    The subset of $\veta$ so obtained is $V \subset \mathbb{R}^L$. The parameterization $\vec{R}_{\vec{x}\vec{x}}(\veta)$ of $\mathcal{R}(\vec{n}, \vec{r})$ is obtained by inverting \eqref{eq:vectorization} for $(\vec{A}(\veta), \vec{B}(\veta), \bPhi(\veta))$ and taking $\vec{R}_{\vec{x}\vec{x}}(\vec{A}(\veta), \vec{B}(\veta), \bPhi(\veta))$.

\section{Identifiability}
\label{sec:identifiability}

    \subsubsection{Definition} For multi-channel factor analysis as defined in Section \ref{sec:model}, we say that an observation covariance matrix $\bSigma_{\vec{x}\vec{x}} \in \mathcal{R}(\vec{n}, \vec{r})$ is \emph{identified} when it can be uniquely decomposed into a sum of appropriately structured components. That is, in the terms of Section \ref{sec:parameterization}, $\bSigma_{\vec{x}\vec{x}}$ is identified if there is a unique triple $(\vec{R}_{\vec{s}\vec{s}}, \vec{R}_{\vec{i}\vec{i}}, \bPhi)$ such that $\vec{R}_{\vec{x}\vec{x}}(\vec{R}_{\vec{s}\vec{s}}, \vec{R}_{\vec{i}\vec{i}}, \bPhi) = \bSigma_{\vec{x}\vec{x}}$. As the covariance matrices $\vec{R}_{\vec{s}\vec{s}}, \vec{R}_{\vec{i}\vec{i}},$ and $\bPhi$  contain all information in MFA about the statistical properties of the signal, interference, and noise respectively, any inference must be based on these covariances.  However, if there exists a different triple $(\tilde{\vec{R}}_{\vec{s}\vec{s}}, \tilde{\vec{R}}_{\vec{i}\vec{i}}, \tilde{\bPhi})$ of structured matrices which also sums to $\bSigma_{\vec{x}\vec{x}}$, then any inference based on the individual values for the summands must be suspect. As $\bSigma_{\vec{x}\vec{x}}$ is not known, practical application of MFA requires that the observation covariance model \eqref{eq:obscov} be guaranteed to be identified using only what is specified $\emph{a priori}$, namely the channel sizes and number of common and {distinct} factors. If the channel sizes $\vec{n}$ and factor numbers $\vec{r}$ permit such a guarantee, we say that MFA is \emph{identifiable} with those channel sizes and factor numbers.
    
    {
    In this definition, identifiability of MFA is a property of the covariance matrix $\bSigma_{\vec{x}\vec{x}}$ and refers to the uniqueness of the second order decomposition \eqref{eq:covariancemodel}, \emph{not} uniqueness of the first order generative model \eqref{eq:stationarymfa_channel}. As discussed in Section \eqref{sec:loadingparam}, the latent factors themselves are not uniquely identifiable even in the noise-free case, as the bases for the common and {distinct} factor spaces can be changed without altering the observations. However, if the MFA covariance  $\bSigma_{\vec{x}\vec{x}}$ is identified in the above sense and a preferred basis for the factor space is chosen, then the uniqueness of the MFA decomposition allows for linear MMSE estimation of the latent factors $\vec{f}$ and $\vec{g}_{c},\ c=1,\dots,C$ in that basis (see \cite[Section IV.B]{ramirez2020} for a related experiment).}

    Identification of $(\vec{R}_{\vec{s}\vec{s}}, \vec{R}_{\vec{i}\vec{i}}, \bPhi)$ from $\bSigma_{\vec{x}\vec{x}}$ breaks into two subproblems, namely \emph{isolation of the idiosyncratic variances} and \emph{separation of the signal and interference covariances}. That is, the former subproblem refers to whether $\bSigma_{\vec{x}\vec{x}}$ uniquely determines $(\vec{R}_{\vec{s}\vec{s}} + \vec{R}_{\vec{i}\vec{i}}, \bPhi)$ while the latter subproblem refers to whether $\vec{R}_{\vec{s}\vec{s}} + \vec{R}_{\vec{i}\vec{i}}$ uniquely determines $(\vec{R}_{\vec{s}\vec{s}}, \vec{R}_{\vec{i}\vec{i}})$.
    
    In Section \ref{sec:uniquedecomp}, conditions on $\vec{n}$ and $\vec{r}$ which resolve the  subproblems and ensure the identifiability of $(\vec{R}_{\vec{s}\vec{s}}, \vec{R}_{\vec{i}\vec{i}}, \bPhi)$ from $\bSigma_{\vec{x}\vec{x}}$ are derived.  The following Section \ref{sec:paramident} investigates the identifiability and associated properties of the parameterization $\vec{R}_{\vec{x}\vec{x}}(\veta)$ of $\mathcal{R}(\vec{n}, \vec{r})$ in terms of $\veta$, 
    which provides the required technical foundation for the asymptotic analysis of Section \ref{sec:asymptotics}. The relationships between the conditions and results of this section are summarized in Figure \ref{fig:conditiondiag}.

    \subsubsection{Generic, Global, and Local Identifiability}
    
     To establish the channel sizes and factor numbers for which MFA is identifiable, it is important to realize that certain degenerate $\bSigma_{\vec{x}\vec{x}}$ will never be identified. For example, choose $\vec{A}, \vec{B}$ such that the block matrix $[\vec{A}\ \vec{B}]$ has some orthogonal rows and all other rows being zero. The resulting $\vec{R}_{\vec{s}\vec{s}} + \vec{R}_{\vec{i}\vec{i}}$ will itself be diagonal and so the noise variances cannot be isolated. Although the subset of $\mathcal{R}(\vec{n}, \vec{r})$ containing the non-identified MFA observation covariance models is not precisely characterized, it is sufficient for practical applications to find conditions on $\vec{n}$ and $\vec{r}$ which imply that such non-identified covariance models are atypical. {In addition, a distinction can be made between \emph{local} and \emph{global} identifiability. A locally identified triple $(\vec{R}_{\vec{s}\vec{s}}, \vec{R}_{\vec{i}\vec{i}}, \bPhi)$ summing to $\bSigma_{\vec{x}\vec{x}}$ is guaranteed to be the unique such triple within some neighborhood, whereas global identifiability extends this guarantee to the entire space.}

     These types of result are common in the identifiability literature; \cite{shapiro1985} and \cite{bekker1997} establish local and global identifiability for single-channel factor analysis in this \emph{generic} sense, while \cite{kiraly2012} examines generic identifiability in low-rank matrix completion. Formally, we call a subset of a $d$-dimensional real vector space $\emph{null}$ if its image under a linear isomorphism to $\mathbb{R}^d$ has Lebesgue measure zero. A statement is \emph{generically} true if it is true for all elements excepting a null subset.

    \begin{figure}
        \centering
        \includegraphics[width=\linewidth]{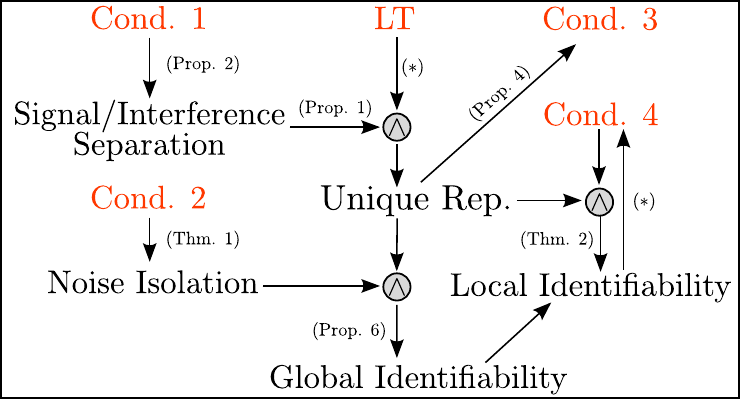}
        \caption{Diagram indicating the relationships between the conditions imposed on the channel sizes and factor numbers and their implications for the different aspects of MFA identification. Directions marked with asterisks were obtained in \cite{ramirez2020}. ``LT" refers to the lower-triangular structure in $\mathbb{A}^*_{L}$ and $\mathbb{B}_L^*$}
        \label{fig:conditiondiag}
    \end{figure}

    \subsubsection{Connections to Previous Results in FA}
     The question of identifiability in single-channel factor analysis was an area of interest for many years. Based on an equation-counting argument, Ledermann \cite{Ledermann1937} provided a heuristic for the maximum number of factors (known as the Ledermann bound). Anderson \cite{anderson1956} set out a simple sufficient condition for identifiability by requiring that the loading matrix contain two disjoint full-rank submatrices after removal of a single row, which will be generically satisfied when the number of common factors is less than half the number of observations. Later, Shapiro \cite{shapiro1985} demonstrated that the Ledermann bound was generically sufficient for \emph{local} identifiability, providing a maximal number of common factors which approaches $n$ rather than $n/2$. Shapiro also conjectured that this identifiability threshold also held for \emph{global} identifiability, which was later shown to be correct by Bekker and ten Berge \cite{bekker1997}.
    
    In this paper, analogous results for multi-channel factor analysis are obtained. The discussion on the identifiability of MFA was opened by \cite[Sec. III]{ramirez2020} and the importance of the problem was recognized. In particular, the authors provided the restrictions which remove rotation invariance used in Proposition \ref{prop:LT}, and give \emph{necessary} conditions on the factor numbers by an equation-counting argument. The authors conjectured that, as in single-channel FA, the threshold obtained by counting knowns and unknowns should be \emph{sufficient} for identifiability, which here is Condition \ref{cond:fac_count}. With the addition of conditions to ensure separability of the signal and interference covariances, which was not treated in \cite{ramirez2020}, this conjecture is verified for \emph{local} identifiability. For \emph{global} identifiability, we instead require the slightly stronger Condition \ref{cond:phisep}.  

    \subsubsection{Identifiability in Related Multi-Channel Methods}
    {Just as classical FA has deep connections with other multivariate statistical methods such as Canonical Correlation Analysis (CCA) \cite{rao1955}, MFA can be related to other techniques for multi-channel data analysis. In particular, CCA (both in classical two-view form \cite{Ibrahim2019} and in the generalized multi-view form \cite{Sorenson2021}) has been successfully used to find latent structures shared across multiple channels, which is also an objective of MFA. Other multi-channel techniques such as Joint Independent Subspace Analysis (JISA) \cite{Lahat2016}, Shared Independent Component Analysis (ShICA) \cite{richard2021} and Deep CCA \cite{Galen2013} also enable discovery of latent structures under differing assumptions on the relations of the shared and unshared aspects to the multi-channel observations. Useful identifiability results for Generalized CCA \cite{Sorenson2021}, Deep CCA \cite{lyu2022},  JISA \cite{Lahat2019}, and ShICA \cite{richard2021} have been obtained through a variety of approaches.}

    {
    However, the MFA-specific identifiability results obtained in this paper are not direct consequences of previous results, and differ in two ways. First, Proposition \ref{prop:genericrotdet} for generic separability of the signal and interference covariances does \emph{not} require that the number of factors $r_0 +r_c$ affecting channel $c$ be less than the channel size, as Condition \ref{cond:rotdet} allows for  $r_0 + r_c$ to be greater than $n_c$ for some channels. If $r_0 + r_c > n_c$, then the latent factors $(\vec{f}, \vec{g}_c)$ cannot be uniquely determined from $\vec{x}_c$ alone, even in the noise-free ($\bPhi = \vec{0}$) case with $\vec{A}_c$ and $\vec{B}_c$ known. Second, the presence of noise in the observations substantially alters the identifiability problem, as unique isolation of the noise variance $\bPhi$ neither implies nor is implied by separability of the signal and interference covariances.}

    \subsection{Identifiability of $\vec{R}_{\vec{x}\vec{x}}(\vec{R}_{\vec{s}\vec{s}}, \vec{R}_{\vec{i}\vec{i}}, \bPhi)$}\label{sec:uniquedecomp}


    \subsubsection{Separation of Signal and Interference Covariances}
    
    The first subproblem of MFA identifiability involves the noise-free part of \eqref{eq:covariancemodel}, namely the combined signal-and-interference covariance $\vec{R}_{\vec{s}\vec{s}}+\vec{R}_{\vec{i}\vec{i}}$. To treat the first subproblem, it is convenient to work with the loading parameterization $\vec{R}_{\vec{x}\vec{x}}(\vec{A}, \vec{B}, \bPhi)$, then relate back to $\vec{R}_{\vec{x}\vec{x}}(\vec{R}_{\vec{s}\vec{s}}, \vec{R}_{\vec{i}\vec{i}}, \bPhi)$. 
    
    Define two equivalence relations on $\mathbb{A} \times \mathbb{B}$ by
    \begin{equation*}
    \begin{aligned}
    (\vec{A}, \vec{B}) \sim_{1} (\tilde{\vec{A}}, \tilde{\vec{B}})\ &\Longleftrightarrow\ \vec{A}\vec{A}^{\tran} + \vec{B}\vec{B}^{\tran} = \tilde{\vec{A}}\tilde{\vec{A}}^{\tran} + \tilde{\vec{B}}\tilde{\vec{B}}^{\tran} \\
     (\vec{A}, \vec{B}) \sim_{2} (\tilde{\vec{A}}, \tilde{\vec{B}})\ &\Longleftrightarrow\ \vec{A}\vec{A}^{\tran}=\tilde{\vec{A}}\tilde{\vec{A}}^{\tran}\ \mbox{and}\ \vec{B}\vec{B}^{\tran}=\tilde{\vec{B}}\tilde{\vec{B}}^{\tran}.
    \end{aligned}
    \end{equation*}
    Under $\sim_1$, two pairs of loading matrices are equivalent if they correspond to the same sum $\vec{R}_{\vec{s}\vec{s}} + \vec{R}_{\vec{i}\vec{i}}$, while under $\sim_2$, two pairs are equivalent if they correspond to the same tuple $(\vec{R}_{\vec{s}\vec{s}}, \vec{R}_{\vec{i}\vec{i}})$. It is clear that $\sim_2$ is a \emph{finer} relation than $\sim_1$, and by definition $\vec{R}_{\vec{s}\vec{s}}  + \vec{R}_{\vec{i}\vec{i}}$ can be uniquely separated into $\vec{R}_{\vec{s}\vec{s}}$ and $\vec{R}_{\vec{i}\vec{i}}$ iff the $\sim_2$-equivalence class associated with $\vec{R}_{\vec{s}\vec{s}}  + \vec{R}_{\vec{i}\vec{i}}$ contains a single $\sim_1$-equivalence class. 

    More can be said about the structure of these $\sim_1$ and $\sim_2$ equivalence classes.  Application of a well-known result (see, e.g., \cite[Lemma 5.1]{Anderson1959}) shows that for pairs $(\vec{A}, \vec{B}), (\tilde{\vec{A}}, \tilde{\vec{B}}) \in \mathbb{A} \times \mathbb{B}$ with  $(\vec{A}, \vec{B}) \sim_1 (\tilde{\vec{A}}, \tilde{\vec{B}})$ we must have
    \[
    \big[ \tilde{\vec{A}}\ \tilde{\vec{B}} \big] = \big[ \vec{A}\  \vec{B}\big ] \vec{Q},
    \]
    for some orthogonal matrix $\vec{Q} \in \mathrm{O}(r_0+r)$. That is, any two $\sim_1$-equivalent pairs are such that the combined loading matrices $[\vec{A}\ \vec{B}]$ and $[\tilde{\vec{A}}\ \tilde{\vec{B}}]$ represent the same map under different orthonormal bases for the  combined factor space of both common and {distinct} factors. Similarly, if $(\vec{A}, \vec{B}) \sim_2 (\tilde{\vec{A}}, \tilde{\vec{B}})$ then $\tilde{\vec{A}} = \vec{A}\vec{Q}_{00}$ for $\vec{Q}_{00} \in \mathrm{O}(r_0)$ and $\tilde{\vec{B}}_{c} = \tilde{\vec{B}}_c \vec{Q}_{cc}$ with $\vec{Q}_{cc} \in \mathrm{O}(r_c)$ for each $c=1,\dots, C$. Partitioning $\vec{Q}$ as,
    \vspace*{-1em}
    \begin{equation}\label{eq:Q_pattern}
       \vec{Q} = \begin{bNiceMatrix}[first-row,last-col=5,nullify-dots]
            r_0 & r_1 & \Cdots & r_C &  \\
            \vec{Q}_{00} & \vec{Q}_{01} & \dots & \vec{Q}_{0C} & r_0 \\
            \vec{Q}_{10} & \vec{Q}_{11} & \dots & \vec{Q}_{1C} & r_1\\
            \vdots& \vdots &  & \vdots  & \Vdots \\
            \vec{Q}_{C0} & \vec{Q}_{C1} & \dots & \vec{Q}_{CC}& r_C\\
        \end{bNiceMatrix}
    \end{equation}
    then $(\vec{A}, \vec{B}) \sim_2 (\tilde{\vec{A}}, \tilde{\vec{B}})$ iff $[\tilde{\vec{A}}\ \tilde{\vec{B}}]$ can be obtained from $[\vec{A}\ \vec{B}]$ by right-multiplication by a \emph{block-diagonal} $\vec{Q}$.

    One distinction between multi-channel FA and single-channel FA with $r_0 + r$ total factors is that not all products $[\vec{A}\ \vec{B}]\vec{Q}$ will correspond to a valid pair of MFA loading matrices $(\tilde{\vec{A}}, \tilde{\vec{B}}) \in \mathbb{A} \times \mathbb{B}$. This is due to the structural zeros of $\mathbb{B}$. If $\vec{Q}$ is block diagonal, the product $[\vec{A}\ \vec{B}]\vec{Q}$ will preserve the structural zeros in $\vec{B}$ and correspond to a valid member of $\mathbb{A} \times \mathbb{B}$. However, the converse is not true without further restrictions on $\vec{n}$ and $\vec{r}$, as non-block-diagonal $\vec{Q}$ which preserve the structural zeros in $\vec{B}$ can exist even in non-degenerate cases. A example of this is given in the Supplementary Materials.

    With the above equivalence relations $\sim_1$ and $\sim_2$, existence of such a non-block-diagonal $\vec{Q}$ occurs exactly when $\vec{R}_{\vec{s}\vec{s}} + \vec{R}_{\vec{i}\vec{i}}$ cannot be uniquely separated. 
    The following proposition gives sufficient conditions on $(\vec{A}, \vec{B})$ so that all $\vec{Q}$ which preserve the structural zeros of $\vec{B}$ are block-diagonal.

    \begin{proposition}\label{prop:submatrices_rotdet}
     For $\vec{A} \in \mathbb{A}$ and $\vec{B}\in\mathbb{B}$, suppose that, after possibly renumbering the channels, the submatrices $\vec{M}_1, \dots, \vec{M}_{C}$ of $[\vec{A}\ \vec{B}]$ have Full Column Rank (FCR), where $\vec{M}_c$ is
    \begin{equation}\label{eq:Mc_def}
    \vec{M}_c = \begin{bmatrix}
    \vec{A}_{<c} & \vec{B}_{<c} \\
    \vec{A}_{>c} & \vec{0}
    \end{bmatrix},
    \ 
    \vec{M}_1 = [\vec{A}^{\tran}_{2}\ \dots\ \vec{A}^{\tran}_{C}]^{\tran}
    \end{equation}
    with $\vec{A}_{<c} = [\vec{A}^{\tran}_1 \dots \vec{A}_{c-1}^{\tran}]^{\tran}$, $\vec{A}_{>c} = [\vec{A}^{\tran}_{c+1} \dots \vec{A}_{C}^{\tran}]^{\tran}$ and $\vec{B}_{<c} = \blkdiag(\vec{B}_1, \dots, \vec{B}_{c-1})$.
    Then any $\vec{Q} \in \mathrm{O}(r_0 + r)$ such that $[\vec{A}\  \vec{B}]\vec{Q} = [\tilde{\vec{A}}\ \tilde{\vec{B}}]$ for some $(\tilde{\vec{A}}, \tilde{\vec{B}}) \in \mathbb{A} \times \mathbb{B}$ must have $\vec{Q}_{ij} = \vec{0}$ for all $i\neq j$ when partitioned as \eqref{eq:Q_pattern}.
    \end{proposition}
     \begin{proof}
     {See Supplementary Materials for proof.}
     \end{proof}
    If the $\sim_1$-equivalence class of $\vec{R}_{\vec{s}\vec{s}} + \vec{R}_{\vec{i}\vec{i}}$ contains any $(\vec{A}, \vec{B})$ which satisfy the condition of Proposition \ref{prop:submatrices_rotdet}, then all elements in the $\sim_1$-equivalence class belong to the same $\sim_2$-equivalence class and so $\vec{R}_{\vec{s}\vec{s}} + \vec{R}_{\vec{i}\vec{i}}$ can be uniquely separated. This can be seen by letting $(\vec{A}, \vec{B})$ satisfy the above condition, so for any $(\tilde{\vec{A}}, \tilde{\vec{B}}) \sim_{2} (\vec{A}, \vec{B})$ the pairs must in fact be related by a block-diagonal orthogonal transformation. Hence, the $\sim_2$ and $\sim_1$ equivalence classes collapse by transitivity. 

    The following condition on $\vec{n}$ and $\vec{r}$ implies that the hypothesis of Proposition \ref{prop:submatrices_rotdet} is generically satisfied on $\mathbb{A} \times \mathbb{B}$, and therefore $\vec{R}_{\vec{s}\vec{s}} + \vec{R}_{\vec{i}\vec{i}}$  can be uniquely separated into $(\vec{R}_{\vec{s}\vec{s}}, \vec{R}_{\vec{i}\vec{i}})$. 
    \begin{condition}\label{cond:rotdet}
        The channel sizes $n_1, \dots, n_C$ and factor numbers $r_0, \dots, r_C$ satisfy $r_{c} \leq n_{c}$ and
        \begin{equation}
        \label{eq:rot_det1}
        \begin{aligned}
        r_0 + r_{<c} &\leq n - n_{c},
        \end{aligned}
        \end{equation}
        for all $c=1,\dots, C$.
    \end{condition}
    {If $r_0 + r_c \leq n_c$ for all channels, then Condition \ref{cond:rotdet} is  satisfied for any ordering of the channels. This follows as $r_0 \leq \min_{c=1,\dots,C} \{n_c - r_c \}$ implies $r_0 \leq (n_1 - r_1) + \dots (n_{c-1} - r_{c-1}) + n_{c+1} + \dots n_{C}$, and so \eqref{eq:rot_det1} is satisfied for all $c=1,\dots,C$. Condition \ref{cond:rotdet} depends on channel ordering, but its use in the following proposition is not order dependent. }
    \begin{proposition}\emph{(Generic Separability of $\vec{R}_{\vec{s}\vec{s}} + \vec{R}_{\vec{i}\vec{i}}$)}\label{prop:genericrotdet}
        If Condition \ref{cond:rotdet} is satisfied for some permutation of the channel numbers, then the subset of $\mathbb{A} \times \mathbb{B}$ which does not satisfy the condition of Proposition \ref{prop:submatrices_rotdet} is null.
    \end{proposition}
     \begin{proof}
     {See Supplementary Materials for proof.}
     \end{proof}

    \subsubsection{Isolation of Noise Variances}
    
    For single-channel FA identifiability, the main criterion is $\phi$ defined for $n,r,\rho \in \mathbb{N}$ as
    \[
    \phi(n, r, \rho) = \frac{r(r+1)}{2} - \frac{\rho(\rho+1)}{2} - \rho(r-\rho) - n.
    \]
    The threshold for global identifiability of single-channel FA \cite{bekker1997} is then $\phi(n, r, 2r - n) > 0$. For MFA, the analogous criterion function $\psi(\vec{n}, \vec{r}, \brho)$ is
    \begin{equation}\label{eq:MFAidentcritfunc}
    \psi(\vec{n}, \vec{r}, \brho) = n + \phi(n\mkern-3mu, r_0, \rho_0) + \sum_{c=1}^C \phi(n_c, r_c, \rho_c) + r_c (r_0 - \rho_0),
    \end{equation}
    with non-negative integer vector $\brho=[\rho_0, \rho_1, \dots, \rho_C]$. The criterion $\psi$ depends on all of the factor numbers $r_0, r_1, \dots, r_C$ and not a function of the total number of factors alone.


    \begin{condition}\label{cond:phisep}
        The channel sizes $\vec{n}$ and factor numbers $\vec{r}$ satisfy $r_c \leq  n_c$ and $r_0 + r \leq n$. In addition, let
        $\psi^*$ be the smallest criterion value over possible MFA reductions,  
    \begin{equation}\label{eq:phisep1}
       \psi^* = \min_{(\vec{n}', \vec{r}', \brho) \in M} \psi(\vec{n}', \vec{r}', \brho) 
    \end{equation}
    where $M \subset \mathbb{N}^{3C+2}$ contains all non-negative $(\vec{n}', \vec{r}', \brho)$ satisfying
    \begin{equation}\label{eq:phisep2}
    \begin{aligned}
        n_c' &\leq n_c, \ c = 1,\dots, C,\\
        r_c' &= (r_c - (n_c - n'_c))_{+}, \ c=1,\dots, C, \\
        r'_0 &=  [r_0 - \textstyle \sum_{c=1}^{C} (n_c - n'_c - r_c)_{+}]_{+}, \\
        \rho_c &\leq \min\{r_c', 2(r_0' + r_c') - n'_c\},\ c = 1,\dots, C, \\
        \rho_0 &= \min\{r_0, 2r_0' + \textstyle \sum_{c=1}^C 2r'_c - \rho_c - n'_c\},  
    \end{aligned}
    \end{equation}
    and $\sum_{c=1}^C n'_c > 0$. Either $\psi^* > 0$ or $M$ is empty.
     \end{condition}

    Similarly to \cite{bekker1997}, $(\vec{A}, \vec{B}, \bPhi)$ is said to have globally identified noise variances if $\vec{R}_{\vec{x}\vec{x}}(\vec{A}, \vec{B}, \bPhi) = \vec{R}_{\vec{x}\vec{x}}(\tilde{\vec{A}}, \tilde{\vec{B}}, \tilde{\bPhi})$ implies that $\bPhi = \tilde{\bPhi}$. The following theorem establishes that Condition \ref{cond:phisep} is sufficient for  $(\vec{A}, \vec{B}, \bPhi)$ to generically have globally identifiable noise variances and hence that the noise variances can be uniquely isolated from $\vec{A}\vec{A}^{\tran} + \vec{B}\vec{B}^{\tran}$.
	\begin{theorem}\emph{(Separation of $\bPhi$)}\label{thm:phi_separation}
    If Condition \ref{cond:phisep} is met, $(\vec{A}, \vec{B}, \bPhi)$ has globally identified noise variances except for a null subset of $\mathbb{A} \times \mathbb{B} \times \mathrm{Diag}_{\geq 0}(n)$.
	\end{theorem}
      \begin{proof}
     {See Appendix for proof.}
     \end{proof}

    If the channel sizes and factor numbers meet both Conditions \ref{cond:rotdet} and \ref{cond:phisep}, the results of this section imply that the observation covariance can be uniquely decomposed into the signal, interference, and noise covariances, except for a null set of degenerate cases. Therefore, interpretation of the individual components of MFA is well-founded. 

    {
    Showing that Condition \eqref{cond:phisep} is sufficient for the unique isolation of the noise variances divides into two cases. The first possibility considered is whether the observation covariance $\vec{R}_{\vec{s}\vec{s}} + \vec{R}_{\vec{i}\vec{i}} + \bPhi$ permits a second representation as $\widetilde{\vec{R}}_{\vec{s}\vec{s}} + \widetilde{\vec{R}}_{\vec{i}\vec{i}} + \widetilde{\bPhi}$ with all noise variances $[\widetilde{\bPhi}]_{ii}$ not equaling $[\bPhi]_{ii},\ i=1,\dots,n$. Such a second representation precludes unique isolation of the noise variances, and implies that $\vec{R}_{\vec{s}\vec{s}} + \vec{R}_{\vec{i}\vec{i}}$ differs by a diagonal matrix from another noise-free MFA covariance with the same factor numbers. This relationship between two noise-free MFA covariances implies the existence of a symmetric matrix $\vec{H}$ of size $r_0 + r$ with appropriate structural zeros and satisfying both an overall rank constraint and rank constraints on the main diagonal blocks, where the constraints are functions of the channel sizes and common and {distinct} factor numbers. As the overall and block rank constraints interact, the integer vector $\brho$ sets the ranks of the diagonal blocks of $\vec{H}$, where the possible values are given in Condition $\eqref{cond:phisep}$ for $n_c' = n_c$, $r_0' = r_0$ and $r_c' = r_c$ for $c=1,\dots, C$. The criterion $\psi$ can be seen \eqref{eq:dimcrit} as the effective number of constraints imposed on $\vec{H}$ minus the degrees of freedom in choosing the diagonal difference matrix. If the criterion $\psi$ is positive for all permitted $\brho$, then all block ranks lead to an overdetermined problem, so generically no such second representation exists.}

    {
    In the second case, the noise variances have $[\widetilde{\bPhi}]_{ii} = [\bPhi]_{ii}$ for $i$ in some index set $\beta$. 
    To resolve the second case, the fact that the difference $[\widetilde{\bPhi}]_{ii} - [\bPhi]_{ii} = 0$ for indices in $\beta$ yields that the associated principal submatrices of $\vec{R}_{\vec{s}\vec{s}} + \vec{R}_{\vec{i}\vec{i}}$ and $\widetilde{\vec{R}}_{\vec{s}\vec{s}} + \widetilde{\vec{R}}_{\vec{i}\vec{i}}$ are equal. Taking the generalized Schur complement in $\vec{R}_{\vec{s}\vec{s}} + \vec{R}_{\vec{i}\vec{i}}$ of this submatrix reduces the second case to the first case with smaller channel sizes $n_c'$ and factor numbers $r_0'$ and $r_c'$ for $c=1,\dots, C$, where the possible reduced channel sizes and factor numbers (for varying index sets $\beta$) are set out in Condition \ref{cond:phisep}. If $\psi$ is positive for all possible reduced $\vec{n}', \vec{r}'$ and the associated possible block ranks $\brho$, then $\bPhi$ can be generically be isolated from $\vec{R}_{\vec{s}\vec{s}} + \vec{R}_{\vec{i}\vec{i}}$.}

    \subsection{Identifiability of $\vec{R}_{\vec{x}\vec{x}}(\veta)$}
    \label{sec:paramident}

    The previous section established conditions under which the MFA decomposition of the observation covariance into the signal, interference, and noise covariance matrices is identifiable and thus interpretable. This section provides complementary results for the identifiability of $\vec{R}_{\vec{x}\vec{x}}(\veta)$. These results are of technical relevance for the analysis of Section \ref{sec:asymptotics} as they allow standard parameter estimation theory to be applied.

    \subsubsection{Unique Representative}\label{sec:uniquerep}
    In constructing the parameterization of $\mathcal{R}(\vec{n}, \vec{r})$ in terms of $\veta$, the first step is defining the subset $\mathbb{A}^*_L \times \mathbb{B}_L^*$ of $\sim_2$-equivalence class representatives. The following proposition establishes that $\mathbb{A}^*_L \times \mathbb{B}_L^*$ contains a unique representative from each $\sim_2$-equivalence class. 

    \begin{proposition}(LT Uniqueness)\label{prop:LT}
    For any $(\vec{A}, \vec{B}) \in \mathbb{A} \times \mathbb{B}$ there is a unique  $(\tilde{\vec{A}}, \tilde{\vec{B}}) \in \mathbb{A}^*_L \times \mathbb{B}^*_{L}$ such that $(\vec{A}, \vec{B}) \sim_2 (\tilde{\vec{A}}, \tilde{\vec{B}})$.
    \end{proposition}
     \begin{proof}
     {See Supplementary Materials for proof.}
     \end{proof}

    This result parallels the use of LT restrictions to select a unique representative loading matrix in single-channe FA. However, for MFA, a previously unrecognized complication occurs when $\vec{R}_{\vec{s}\vec{s}} + \vec{R}_{\vec{i}\vec{i}}$ cannot be uniquely separated. In this case, there exist multiple elements of $\mathbb{A}^*_L \times \mathbb{B}^*_{L}$ which are $\sim_1$ equivalent, and so the LT restrictions do not select a unique representative from each $\sim_1$-equivalence class. 
    
     Proposition \ref{prop:rot_det_nesc} applies a result for confirmatory factor analysis \cite{Bekker1986} to give a \emph{necessary} condition that the channel sizes and factor numbers must satisfy so that the LT restriction will distinguish a unique representative of the $\sim_1$-equivalence class. Condition \ref{cond:rotdet} is \emph{sufficient} for the same result.

    \begin{condition}\label{cond:rot_det_nesc}
    The  channel sizes $\vec{n}$ and the factor numbers $\vec{r}$ satisfy
    \begin{equation}\label{eq:rot_det_nesc}
    \textstyle r_0r + \sum_{c=1}^C r_c r_{<c} \leq \sum_{c=1}^{C} (n - n_c)r_c.
    \end{equation}
    \end{condition}
    
    \begin{proposition}\label{prop:rot_det_nesc}
    If almost all $(\tilde{\vec{A}}, \tilde{\vec{B}}) \in \mathbb{A}^*_L \times \mathbb{B}^*_L$ are the unique representative in $\mathbb{A}^*_L \times \mathbb{B}^*_L$ of their $\sim_1$-equivalence class, then Condition \ref{cond:rot_det_nesc} is satisfied. Conversely, if Condition \ref{cond:rotdet} is satisfied, then almost all $(\tilde{\vec{A}}, \tilde{\vec{B}})$ are the unique representative in $\mathbb{A}^*_L \times \mathbb{B}^*_L$ of their $\sim_1$-equivalence class.
    \end{proposition}
    \begin{proof}
     {See Supplementary Materials for proof.}
     \end{proof}


    In connection to $\vec{R}_{\vec{x}\vec{x}}(\veta)$, if $\veta, \tilde{\veta} \in V$  are obtained by \eqref{eq:vectorization}  applied to $(\vec{A}, \vec{B}, \bPhi_0)$ and $(\tilde{\vec{A}}, \tilde{\vec{B}}, \bPhi_0)$  respectively, for $(\vec{A}, \vec{B})$ and $(\tilde{\vec{A}}, \tilde{\vec{B}})$ in $\mathbb{A}^*_L \times \mathbb{B}_{L}^*$, then $\vec{R}_{\vec{x}\vec{x}}(\veta) = \vec{R}_{\vec{x}\vec{x}}(\tilde{\veta})$ implies $\veta = \tilde{\veta}$ except on a null subset of $V$. That is, if the noise variances are known, then Condition \ref{cond:rotdet}, under which the signal and interference covariances can generically be separated, also yields that $\veta$ is generically globally identifiable. The following two subsections treat the typical case where $\bPhi$ is unknown.


    \subsubsection{Local Identifiability}\label{sec:localident}
    For fixed $\bSigma_{\vec{x}\vec{x}} \in \mathcal{R}(\vec{n},\vec{r})$, the equation $\vec{R}_{\vec{x}\vec{x}}(\veta) = \bSigma_{\vec{x}\vec{x}}$ defines a quadratic system of equations in the entries of $\veta$. As this system is nonlinear, simply counting the number of knowns in $\bSigma_{\vec{x}\vec{x}}$ and the number of unknowns in $\veta$ is not sufficient to determine whether a solution is unique. However, linearization of the system by considering the differential $d\vec{R}_{\vec{x}\vec{x}}(\veta)$ allows investigation of \emph{local identification}. Here, local identification at $\veta$ means that there is a neighborhood of $\veta$ on which $\vec{R}_{\vec{x}\vec{x}}(\veta)$ is an invertible map. We say that MFA is \emph{generically locally identifiable} with channel sizes $\vec{n}$ and factor numbers $\vec{r}$ if almost all $\veta \in V$ are locally identified. 

    As $\vec{R}_{\vec{x}\vec{x}}(\veta)$ is a smooth map, local identification at $\veta$ follows by the inverse function theorem if $d\vec{R}_{\vec{x}\vec{x}}(\veta)$ is injective. To assess this, a key condition follows from the tabulation of knowns and unknowns in MFA with channel sizes $\vec{n}$ and factor numbers $\vec{r}$, as discussed in \cite[Sec. III]{ramirez2020}.
    \begin{condition}\label{cond:fac_count}
      The number of common factors $r_0$ satisfies
    \begin{equation}
        \label{eq:local_identifiability_bound1}
        r_0 \leq \textstyle\frac{1}{2}\left(2n + 1 - \sqrt{8(n+D)+1} \right),
        \end{equation}
        where 
        \begin{equation}
        \label{eq:K1def}
        D = \textstyle\sum_{c=1}^C n_c r_c - \frac{1}{2}r_c(r_c-1)
        \end{equation}
        and for each channel $c=1,\dots, C$, the number of {distinct} factors in that channel satisfies
        \begin{equation}
        \label{eq:local_identifiability_bound2}
        r_c \leq \textstyle\frac{1}{2} \big(2n_c + 1 - \sqrt{8n_c+1} \big).
        \end{equation}
    \end{condition} 
    The following proposition, which was proven in \cite{ramirez2020}, shows that Condition \ref{cond:fac_count} is \emph{necessary} for $d\vec{R}_{\vec{x}\vec{x}}$ to be injective. 

    \begin{proposition}\label{prop:local_ident_nesc}
    It is \emph{necessary} that the channel sizes $\vec{n}$ and factor numbers $\vec{r}$  satisfy Condition \ref{cond:fac_count} for $d\vec{R}_{\vec{x}\vec{x}}(\veta)$ to be injective at any $\veta \in V$.
    \end{proposition}

    However, ensuring that $d\vec{R}_{\vec{x}\vec{x}}(\veta)$ is \emph{generically} injective is more challenging, as it requires examining the differential itself in addition to the dimensions of the domain and codomain. The following theorem shows that, when combined with the separability result of Proposition \ref{prop:genericrotdet}, Condition \ref{cond:fac_count} is also \emph{sufficient} for local identifiability. 

    \begin{theorem}\emph{(Local Identifiability)}\label{prop:local_identifiability}
        If  the channel sizes $\vec{n}$ and factor numbers $\vec{r}$ satisfy Conditions \ref{cond:rotdet} and \ref{cond:fac_count}, then the differential $d\vec{R}_{\vec{x}\vec{x}}(\veta)$ is generically injective.
    \end{theorem}
     \begin{proof}
     {See Appendix for proof.}
     \end{proof}

    \subsubsection{Global Identifiability}\label{sec:globalident}
 Although the local identifiability result of Theorem \ref{prop:local_identifiability} provides valuable information about the behavior of $\vec{R}_{\vec{x}\vec{x}}(\veta)$ on small neighborhoods and will be needed in Section \ref{sec:asymptotics}, a stronger global identifiability result for $\vec{R}_{\vec{x}\vec{x}}(\veta)$ is desired.  The following proposition combines the results of Section \ref{sec:uniquerep} with Proposition \ref{prop:genericrotdet} and Theorem \ref{thm:phi_separation} to show that $\vec{R}_{\vec{x}\vec{x}}(\veta)$ is an invertible map, excepting a null set of $\veta$.
 
    \begin{proposition}\emph{(Global Identifiability)}\label{prop:global_identifiability}
        If the channel sizes $\vec{n}$ and factor numbers $\vec{r}$ satisfy Conditions \ref{cond:rotdet} and \ref{cond:phisep}, then there exists a subset $\tilde{V} \subset V$ such that $\vec{R}_{\vec{x}\vec{x}}(\veta)$ is injective on $\tilde{V}$ and $V \setminus \tilde{V}$ is null.
    \end{proposition}
	     \begin{proof}
     {See Supplementary Materials for proof.}
     \end{proof}

\section{Asymptotics}
\label{sec:asymptotics}

	\subsection{Estimation}\label{sec:estimation}
	
	Suppose $T$ observation vectors $\vec{x}_{1}, \dots, \vec{x}_T$ are available and are i.i.d. with covariance $\bSigma_{\vec{x}\vec{x}} = \vec{R}_{\vec{x}\vec{x}}(\mathring{\veta})$. In this setting, \cite{ramirez2020} presents an estimation procedure to obtain the value of $\veta$ which maximizes the likelihood of the observations under the assumption that the latent factors and idiosyncratic errors are jointly multivariate normal. Under those distributional assumptions, the implied log density for $\vec{x}$ is
    \begin{equation}
    \label{eq:logdensity}
    \begin{aligned}
    \log f(\vec{x}; \veta) = & - \frac{1}{2} \log \det \vec{R}_{\vec{x}\vec{x}}(\veta) - \frac{1}{2} \vec{x}^\tran \vec{R}_{\vec{x}\vec{x}}^{-1}(\veta) \vec{x} + K.
    \end{aligned}
    \end{equation}
    With this density, estimation of $\vec{A}, \vec{B}, \bPhi$ from $\vec{x}_{1}, \dots \vec{x}_T$ is framed as the optimization problem 
    \begin{equation}
    \label{eq:optimization}
    \min_{\veta \in {V}}\ \log \det \vec{R}_{\vec{x}\vec{x}}(\veta) + \tr\vec{R}_{\vec{x}\vec{x}}^{-1}(\veta) \vec{S}_T
    \end{equation}
    where $\vec{S}_T$ is the sample covariance,
    $
    \vec{S}_T = T^{-1}\textstyle \sum_{j=1}^T \vec{x}_j \vec{x}_j^{\tran}, 
    $
    and the sample objective function $\ell_{T}$ is
    \begin{equation}
    \label{eq:sampleobjective}
    \ell_{T}(\vec{S}_T ; \veta) \equiv \log \det \vec{R}_{\vec{x}\vec{x}}(\veta) + \tr\vec{R}_{\vec{x}\vec{x}}^{-1}(\veta) \vec{S}_T.
    \end{equation}

    To avoid the Heywood cases \cite{Lawley1962}, we will restrict attention to $(\vec{A}, \vec{B}, \bPhi)$ such that $\min_{i} [\bPhi]_{ii} \geq \epsilon$ for some fixed $\epsilon > 0$. This has the advantage of ensuring that the smallest eigenvalue $\vec{R}_{\vec{x}\vec{x}}(\vec{A}, \vec{B}, \bPhi)$ is bounded away from zero. Letting $V' \subset V$ contain all $\veta$ which satisfy this additional requirement, we define
     the estimators $\hat{\vec{A}}_{T}$, $\hat{\vec{B}}_{T}$, $\hat{\bPhi}_T$ as those obtained from the minimizer of $\ell_{T}(\vec{S}_T ; \veta)$, 
    \begin{equation}
    \label{eq:gmle}
    \hat{\veta}_{T} = \argmin_{\veta \in {V}'}  \ell_{T}(\vec{S}_T ; \veta).
    \end{equation}
    The estimators for the MFA parameters, $\hat{\vec{A}}_{T}$, $\hat{\vec{B}}_{T}$, $\hat{\bPhi}_T$, are obtained by inverting \eqref{eq:vectorization} for $\hat{\veta}_{T}$.

    As the latent factors and idiosyncratic errors are not observed, the assumption of joint multivariate normality can be difficult to support. Therefore, the asymptotic results of Section \ref{sec:asympprops} are obtained by treating \eqref{eq:sampleobjective} as a \emph{quasi-loglikelihood}\cite{Heyde1997} objective function to be optimized, rather than requiring that \eqref{eq:logdensity} be the true likelihood. The results on the asymptotic consistency and normality of the estimators do not require the joint normality of the latent factors and errors. These results instead require only mild moment assumptions, so the estimators are asymptotically valid if the latent vectors are non-normal.

\subsection{Asymptotic Properties}\label{sec:asympprops}

In this section it is primarily assumed that the observation vectors $\vec{x}_1, \dots, \vec{x}_T$ are independent and identically distributed with mean zero and MFA covariance model \eqref{eq:covariancemodel},
\begin{equation}
\Var(\vec{x}_1) = \bSigma_{\vec{x}\vec{x}} \equiv \vec{R}_{\vec{x}\vec{x}}(\mathring{\vec{A}}, \mathring{\vec{B}}, \mathring{\bPhi}).
\end{equation}  
The true values $\mathring{\vec{A}}$ and $\mathring{\vec{B}}\equiv \blkdiag(\mathring{\vec{B}}_{1}, \dots \mathring{\vec{B}}_{C})$ are such that $(\mathring{\vec{A}}, \mathring{\vec{B}}) \in \mathbb{A}_L^* \times \mathbb{B}_L^*$ and $[\mathring{\bPhi}]_{ii} > \epsilon$ for all $i=1,\dots,n$. The vectorization \eqref{eq:vectorization} of $(\mathring{\vec{A}}, \mathring{\vec{B}}, \mathring{\bPhi})$ is $\mathring{\veta} \in V'$. The higher moments of $\vec{x}_1$ are not specified, and in particular the observations need not be normally distributed. {Theorem \ref{thm:consistency} also speaks to the misspecified case where $\bSigma_{\vec{x}\vec{x}} \notin \mathcal{R}(\vec{n}, \vec{r})$.}

 
The next theorem shows that identification of the true $\bSigma_{\vec{x}\vec{x}}$ is enough to ensure that the estimators are consistent for the factor loading parameters $\vec{A}$ and $\vec{B}$ and the idiosyncratic noise variance $\bPhi$. This follows from the fact that $\vec{x}$ has finite second moment and the exclusion of singular covariance models in the definition of the parameter space. The objective function $\ell_T$ is then sufficiently well-behaved so that the maximizer $\hat{\veta}_{T}$ of $\ell_T$ converges to the maximizer of $\ell_0 \equiv E[\ell_T]$, which is $\mathring{\veta}$. Convergence of  $\hat{\vec{A}}_T, \hat{\vec{B}}_T, \hat{\vec{\bPhi}}_T$ to $\mathring{\vec{A}}, 
\mathring{\vec{B}}, \mathring{\bPhi}$ then follows.

\begin{theorem} \emph{(Consistency)}\label{thm:consistency}
Suppose $\vec{x}_1, \vec{x}_2, \dots$ are an i.i.d. sequence of random vectors with $E[\vec{x}_1] = \vec{0}$ and positive definite $\Var(\vec{x}_1) \equiv \bSigma_{\vec{x}\vec{x}}$. {If there exists a unique $\mathring{\vec{R}}_{\vec{x}\vec{x}} \in \mathcal{R}(\vec{n},\vec{r})$ minimizing $D_{KL}\big(\mathcal{N}(\vec{0}, \bSigma_{\vec{x}\vec{x}})||\ \mathcal{N}(\vec{0}, \mathring{\vec{R}}_{\vec{x}\vec{x}})\big)$ with $\mathring{\vec{R}}_{\vec{x}\vec{x}} = \vec{R}_{\vec{x}\vec{x}}(\mathring{\veta})$ for $\mathring{\veta} \in V'$ in the interior of the globally identified set $\tilde{V}$ defined in Proposition \ref{prop:global_identifiability}, then $\hat{\vec{A}}_T, \hat{\vec{B}}_T, \hat{\vec{\bPhi}}_T$ converge in probability to $\mathring{\vec{A}}, \mathring{\vec{B}}, \mathring{\bPhi}$ respectively.}
\end{theorem}
     \begin{proof}
     {See Appendix for proof.}
     \end{proof}
{In particular, if the model is correctly specified with $\Var(\vec{x}_1) = \vec{R}_{\vec{x}\vec{x}}(\mathring{\veta})$ for $\mathring{\veta}$ in the interior of the globally identified set, then the estimators $\hat{\vec{A}}_{T}, \hat{\vec{B}}_T, \hat{\bPhi}_{T}$ are consistent.}

The following theorem shows that the estimators have a limiting Gaussian distribution when the observation distribution has a finite fourth moment. This is obtained from consistency of the estimators and the nature of $\ell_T$ in a neighborhood of the covariance $\vec{R}_{\vec{x}\vec{x}}(\mathring{\veta})$. In particular, Theorem \ref{prop:local_identifiability} is used to show that the objective generically has a positive second differential, and so the limiting covariance is positive definite. {As the limiting distribution is non-degenerate, $\hat{\veta}_{T} - \mathring{\veta}$ converges to zero in probability at the standard parametric rate of $T^{-1/2}$.}

\begin{theorem}\emph{(Asymptotic Normality)}\label{thm:normality}
Assume the conditions of Theorem \ref{thm:consistency} are satisfied {with $\bSigma_{\vec{x}\vec{x}} = \vec{R}_{\vec{x}\vec{x}}(\mathring{\veta})$}. In addition, assume that $\vec{x}_1$ satisfies $E[||\vec{x}_1||^4] < \infty$. Under these conditions, the estimated parameters $\hat{\veta}_T$ converges in distribution as
\begin{equation}
\sqrt{T}(\hat{\veta}_T - \mathring{\veta}) \overset{d}{\to} \mathcal{N}(\vec{0}_{L}, \vec{W})
\end{equation}
for positive-definite matrix $\vec{W}$ in the Appendix, \eqref{eq:asymptoticcov}.
\end{theorem}
     \begin{proof}
     {See Appendix for proof.}
     \end{proof}

\section{Experiments}\label{sec:experiments}
\subsection{Numeric Comparison of Conditions}

  \begin{figure}
        \centering
        \includegraphics[width=\linewidth]{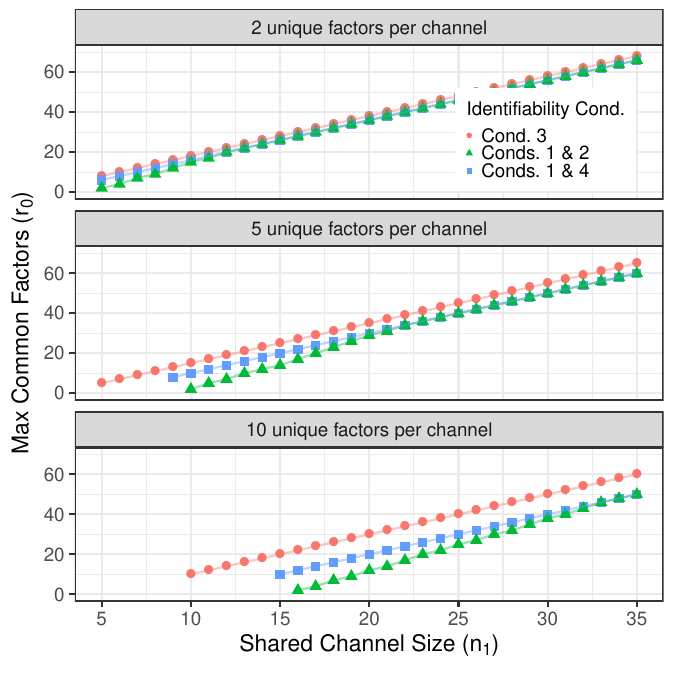}
        \caption{{Comparison of maximum common factor number $r_0$ under three identifiability conditions, for varying channel sizes. Channel structure depicted is three equally-sized channels, with $r_c = 2, 5, 10$ {distinct} factors for $c=1,2,3$.}}
        \label{fig:condcomparison}
    \end{figure}

    {
    In Section \ref{sec:identifiability}, conditions on the channel sizes and factor numbers and their implications for MFA identifiability are given. To gain intuition for how varying channel sizes and factor numbers affects the satisfaction of these conditions,  Figure \ref{fig:condcomparison} depicts the illustrative case of three equal-size channels.
    The figure compares Condition \ref{cond:rot_det_nesc}, which is \emph{necessary} for identifiability, to the hypotheses of Theorem \ref{prop:local_identifiability} and Proposition \ref{prop:global_identifiability} which are respectively \emph{sufficient} for generic local and global identifiability. To interpret Figure \ref{fig:condcomparison}, examine channel size $n_1 = 15$ in the middle panel with $r_c=5$. In this case, Proposition \ref{prop:local_ident_nesc} implies that, for $r_0 > 25$, the set of $\sim_2$-representatives $\mathbb{A}_L^* \times \mathbb{B}_L^* $ does not contain unique representatives of almost all $\sim_1$-equivalence classes, preventing the separation of signal and interference. This is indicated by the circle at $n_1=15$ and $r_0=25$. Further, $19$ is the maximum $r_0$ which guarantees local identifiability under Theorem \ref{prop:local_identifiability} as shown by the square at $n_1=15$ and $r_0=19$. Finally, $14$ is the largest $r_0$ which yields global identifiability under Proposition \ref{prop:global_identifiability}, which is indicated by the triangle at $n_1=15$ and $r_0=14$. The maximum $r_0$ for which generic global and local identifiability can be respectively guaranteed under Proposition \ref{prop:global_identifiability} and Theorem \ref{prop:local_identifiability} agree as channel size increases, but the channel size for which the condition agree increases as the distinct factor number increases. In addition, the gap between Condition \ref{cond:rot_det_nesc} (which is necessary for identifiability) and Conditions \ref{cond:rotdet} \& \ref{cond:fac_count} is constant in the shared channel size and small relative to the total factor number $r_0 + r$. Although the results of this paper give only sufficient conditions for local and global identifiability, the experiment in this section demonstrates that these conditions are close to Condition \ref{cond:rot_det_nesc} which is an upper-bound for MFA identifiability.
    For further discussion and comparisons with unequal channels, see \cite{Stanton2023}. }

\subsection{Asymptotic Behavior of Estimators}

 {
 To verify the consistency of $\hat{\veta}_{T}$ resulting from Theorem \ref{thm:consistency}, 
 Figure \ref{fig:consistency} shows the Normalized Mean Square Error (NMSE) $||\hat{\veta}_{T} - \mathring{\veta}||^2 / ||\mathring{\veta}||^2$ when the model is correctly specified and the channel sizes and factor numbers satisfy Conditions 1 \& 2. For each trial, non-zero entries of the true parameters $\mathring{\vec{A}} \in \mathbb{A}^*_L, \mathring{\vec{B}} \in\mathbb{B}^*_L$ and $\mathring{\bPhi}$ are independent $\mathcal{N}(0, 1)$ samples. For entries constrained to be non-negative, the absolute value is taken. Initial values for the estimation procedure of \cite{ramirez2020} are independently obtained in the same fashion, and $\hat{\veta}_{T}$ is computed from $T$ independent samples with covariance $\bSigma_{\vec{x}\vec{x}}$. }

 {
 For $C=3$ channels with $n_c=8$, Figure \ref{fig:condcomparison} shows that the largest $r_0$ meeting Conditions 1 \& 2 is $r_0 = 9$ while the largest $r_0$ meeting Conditions 1 \& 4 is $r_0 = 12$. The decreasing NMSE in Figure \ref{fig:consistency} for $r_0 = 9$ verifies that the parameters $\veta$ can be consistently estimated when global identifiability is guaranteed, while the non-decreasing NMSE for $r_0 = 12$ shows that local identifiability alone is insufficient for consistency. The decreasing NMSE for intermediate cases $9 < r_0 < 12$ may indicate that MFA is globally identifiable for those factor numbers, but this is not given by Proposition \ref{prop:global_identifiability}.}

  \begin{figure}
        \centering
        \includegraphics[width=\linewidth]{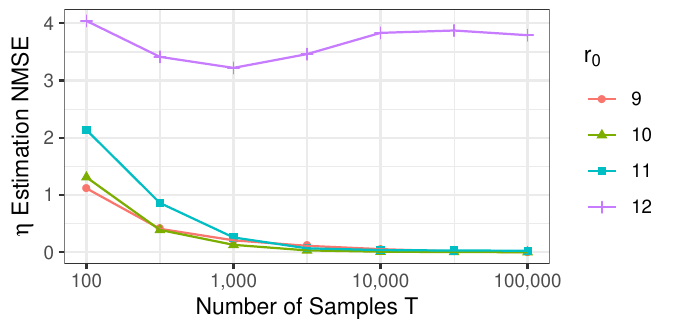}
        \caption{{Experimental validation of Theorem \ref{thm:consistency}  for varying common factor number $r_0$, with $C=3$ channels of size $n_c=8$ and $r_c=2$ {distinct} factors, $c=1,2,3$. Points indicate average NMSE from $1000$ Monte Carlo trials at each setting.  }}
        \label{fig:consistency}
    \end{figure}

\section{Discussion}
\label{sec:discussion}

This paper provides a set of theoretical results for multi-channel factor analysis, which justify applying MFA to analyze the second-order structure of multi-channel observations. Conditions on the allowable number of common and {distinct} factors which guarantee generic uniqueness of the decomposition of the covariance into across-channel, within-channel, and idiosyncratic components are set out in Section \ref{sec:identifiability}. These identifiability results ensure that conclusions drawn from MFA are meaningful as long as the channel sizes and factor numbers satisfy the appropriate conditions. Further, although the estimation procedure proposed in \cite{ramirez2020} is obtained by likelihood maximization under the assumption of normality for the latent vectors, the results of Section \ref{sec:asymptotics} demonstrate that violation of this assumption does not affect the asymptotic validity of the resulting estimators.


 When introducing multi-channel factor analysis, \cite{ramirez2020} discusses the broad potential applicability of the MFA model to diverse problems in signal processing, statistics, and machine learning where channel structure is a relevant feature. The promise of this method comes from the utility of the decomposition of the observation covariance into structured parts corresponding to the latent signal, interference and noise, the uniqueness of which can now be verified. {The identifiability results of this paper are obtained assuming that the signal and interference dimensions are prespecified. Many applications of interest would require estimating these dimensions from the observations, which is a challenging order selection problem. In single-channel FA, techniques such as maximization of an information criterion \cite{Akaike1987}, bi-cross-validation \cite{Owen2016}, or eigenvalue analysis \cite{Ahn2009} can be used to estimate the number of common factors. Adapting these techniques for order selection in MFA is an important direction for future work.}

 \section*{Acknowledgment}\label{sec:ack}
 The authors thank the reviewers for their constructive comments.
 This work was supported in part by National Science Foundation grants DMS-1923142, CNS-1932413, and DMS-2123761. The work of I. Santamaria was funded by AEI /10.13039/501100011033 and FEDER UE under grant PID2022-137099NB-C43 (MADDIE). The work of D. Ram{\'\i}rez was partially supported by MICIU/AEI/10.13039/501100011033/FEDER, UE, under grant PID2021-123182OB-I00 (EPiCENTER), by the Office of Naval Research (ONR) Global under contract N62909-23-1-2002, and by the Spanish Ministry of Economic Affairs and Digital Transformation and the European Union-NextGenerationEU through the UNICO 5G I+D SORUS project.

\bibliographystyle{ieeetr}
{\footnotesize
\bibliography{mainfinalfinalfinal}
}
\setcounter{secnumdepth}{0}
\begin{appendices}

\section*{Appendix \\ Proofs of Theorems $1-4$}\label{app:ident}
    For ease of notation, let $\mathbb{A} \oplus \mathbb{B} \subset \mathbb{R}^{n \times (r_0 + r)}$ be the subspace of matrices which can be written as $[\vec{A}\ \vec{B}]$ for some $\vec{A} \in \mathbb{A}, \vec{B} \in \mathbb{B}$. The spaces $\mathbb{A} \oplus \mathbb{B}$ and $\mathbb{A} \times \mathbb{B}$ are trivially isomorphic. Further, let $(\mathbb{A} \oplus \mathbb{B})^*$ contain all FCR elements of $\mathbb{A} \oplus \mathbb{B}$. As long as $r_0 + r \leq n$ and $r_c \leq n_c$ for each $c$, $(\mathbb{A} \oplus \mathbb{B})^*$ is an open submanifold. The set $\mathbb{A}_L \oplus \mathbb{B}_L$ is defined similarly. 
    
    As many of the propositions proved here involve null sets, we distinguish between null subsets of the unrestricted loadings $\mathbb{A} \times \mathbb{B}$ and null subsets of the equivalence class representatives $\mathbb{A}_L^* \times \mathbb{B}_L^*$. A null subset of  $\mathbb{A} \times \mathbb{B}$ need not correspond to a null subset of representatives. For example, $\mathbb{A}_L^* \times \mathbb{B}_L^*$ is itself null in $\mathbb{A} \times \mathbb{B}$. Lemma \ref{lem:nullequiv} connects the two notions.
    Proofs of the following lemmas can be found in the Supplementary Materials for this paper.

    \begin{lemma}\label{lem:nullequiv}
    If $\mathcal{C} \subset \mathbb{A} \times \mathbb{B}$ is a null set which is a union of $\sim_2$-equivalence classes, then the set of representatives $\tilde{C} \subset \mathbb{A}^*_L \times \mathbb{B}^*_L$ is null in  $\mathbb{A}_{L} \times \mathbb{B}_L$.
    \end{lemma}

\begin{lemma}\label{LT_lemma}
    Let $\vec{H} = [\vec{H}^{\tran}_{1}\ \vec{H}^{\tran}_{2}]^{\tran}$ and $\vec{Z} = [\vec{Z}^{\tran}_{1}\ \vec{Z}^{\tran}_{2}]^\tran$ be $m \times p$ matrices with $\vec{H}_1,\vec{Z}_1 \in \mathbb{R}^{p \times p}$. If $\vec{H}_1$ is invertible, $\vec{H}_1$ and $\vec{Z}_1$ are LT, then $\vec{H}\vec{Z}^{\tran} + \vec{Z}\vec{H}^{\tran} = \vec{0}$ implies $\vec{Z} = \vec{0}$.
    \end{lemma}

\begin{lemma}(Maximal Rank Sub-Loadings)\label{lem:maxranksubload}
    Let $(\mathbb{A} \oplus \mathbb{B})^{**}$ be the subset of $(\mathbb{A} \oplus \mathbb{B})^{*}$ containing FCR $[\vec{A}\ \vec{B}]$ where, for all $c=1,\dots, C$, the rank of all submatrices of the channel $c$ loadings $[\vec{A}_c\ \vec{B}_c]$ are maximal. That is, if $\vec{D}$ is any $s \times t$ submatrix of $[\vec{A}_c\ \vec{B}_c]$, then
    $\rank(\vec{D}) = \min\{s, t\}$.
     The complement $\mathbb{A} \oplus \mathbb{B} \setminus (\mathbb{A} \oplus \mathbb{B})^{**}$ is null.
    \end{lemma}
 \begin{proof}[Proof of Theorem 1]\label{sec:thm1proof}
    Define $\mathcal{B} \subset \mathbb{A} \times \mathbb{B} \times \mathrm{Diag}_{\geq 0}(n)$ as the subset where $(\vec{A}, \vec{B}, \bPhi)$ does not have globally identified noise variances. For any $\bPhi' \succeq \bPhi$, if $(\vec{A}, \vec{B}, \bPhi)$ is in $\mathcal{B}$ then so too is $(\vec{A}, \vec{B}, \bPhi')$ as $\vec{R}_{\vec{x}\vec{x}}(\vec{A}, \vec{B}, \bPhi) = \vec{R}_{\vec{x}\vec{x}}(\tilde{\vec{A}}, \tilde{\vec{B}}, \tilde{\bPhi})$ implies $\vec{R}_{\vec{x}\vec{x}}(\vec{A}, \vec{B}, \bPhi + [\bPhi' - \bPhi]) = \vec{R}_{\vec{x}\vec{x}}(\tilde{\vec{A}}, \tilde{\vec{B}}, \tilde{\bPhi} + [\bPhi' - \bPhi])$ with $\bPhi' \neq \tilde{\bPhi} + [\bPhi' - \bPhi]$. As the set of $\bPhi'$ greater than $\bPhi$ has positive Lebesgue measure in $\mathrm{Diag}(n)$, $\mathcal{B}$ is null in $ \mathbb{A} \times \mathbb{B} \times \mathrm{Diag}_{\geq 0}(n)$ iff its projection onto $\mathcal{A} \times \mathcal{B}$ is null. That this projection is null is shown in the remainder of the proof.
    
    Let $\mathcal{U} \subset \mathbb{A}\oplus \mathbb{B}$ contain those $[\vec{A}\ \vec{B}] $ such that there exists some other $[\tilde{\vec{A}}\ \tilde{\vec{B}}]$ and a diagonal $\bvarphi = \bPhi - \tilde{\bPhi}$ with
    \begin{align}\label{eq:nonseplowrank}
    \vec{A}\vec{A}^{\tran} + \vec{B}\vec{B}^{\tran} + \bvarphi &= \tilde{\vec{A}}\tilde{\vec{A}}^{\tran} + \tilde{\vec{B}}\tilde{\vec{B}}^{\tran},
    & \bvarphi &\neq \vec{0}.
    \end{align}
    If $[\vec{A}\ \vec{B}] \in \mathcal{U}$ for $\bvarphi = \bPhi - \tilde{\bPhi} \neq 0$, then $(\vec{A}, \vec{B}, \bPhi)$ is in $\mathcal{B}$, and so showing $\mathcal{U}$ to be null will imply that $\mathcal{B}$ is null.
    
    To show $\mathcal{U}$ is null under the conditions of Theorem \ref{thm:phi_separation}, $\mathcal{U}$ is partitioned into a number of cases. Let $\mathcal{U}^*$ contain the elements of $\mathcal{U}$ satisfying the maximal rank condition of Lemma \ref{lem:maxranksubload}, $\mathcal{U}^* = (\mathbb{A} \oplus \mathbb{B})^{**} \cap \mathcal{U}$.  Next, note that $\mathcal{U}^*$ can be written as the finite union of $\mathcal{U}^*_{\beta}$ where $\beta \subset \{1,\dots,n\}$ is a proper subset of the possible indices. The subset $\mathcal{U}_{\beta}$ is obtained by adding the restriction that $[\bvarphi]_{ii} = 0$ for all $i\in\beta$ and $[\bvarphi]_{jj} \neq 0$ for $j \in \beta^c$ to \eqref{eq:nonseplowrank}. The proof proceeds in two steps. In the first step, $\beta$ is the empty set and so $\bvarphi$ is non-singular. Results from differential geometry will imply that $\mathcal{U}^*_{\emptyset}$ is null. In the second step, the diagonal of $\bvarphi$ has zeros, which enables reduction to the invertible case with smaller $\vec{n}$ and $\vec{r}$. 
    \subsubsection*{Case $1$: $\varphi$ non-singular}
    In the primary case, $\beta = \emptyset$ and so $\bvarphi$ is non-singular. To eliminate the quantification over $\tilde{\vec{A}}$ and  $\tilde{\vec{B}}$ in the definition of $\mathcal{U}$, construct the block matrices
    \[
    \vec{M} = \begin{bmatrix}
        \bvarphi & \vec{A} & \vec{B} \\
        \vec{A}^{\tran} & -\vec{I}_{r_0} & \vec{0} \\
        \vec{B}^{\tran} & \vec{0} & -\vec{I}_{r} \\
    \end{bmatrix},\quad 
    \vec{M}_c = 
    \begin{bmatrix}
        \bvarphi_{c} & \vec{A}_{c} & \vec{B}_c \\
        \vec{A}^{\tran}_{c} & -\vec{I}_{r_0} & \vec{0} \\
        \vec{B}^{\tran}_{c} & \vec{0}  & -\vec{I}_{r_c}
    \end{bmatrix},
    \]
    for $c=1,\dots, C$, where $\bvarphi_c$ is the submatrix of $\bvarphi$ in the $c$th channel. Additivity of rank with respect to the Schur complement implies
    \begin{equation}\label{eq:rankeqns}
    \begin{aligned}
    r_0 + r + \rank(\tilde{\vec{A}}\tilde{\vec{A}}^{\tran} + \tilde{\vec{B}}\tilde{\vec{B}}^{\tran}) = n + \rank(\vec{H}), \\
    r_0 + r_c + \rank(\tilde{\vec{A}}_c\tilde{\vec{A}}_c^{\tran} + \tilde{\vec{B}}_c\tilde{\vec{B}}_c^{\tran}) = n_c + \rank(\vec{H}_c), \\
    \end{aligned}
    \end{equation}
    for $c=1,\dots, C$, where $\vec{H}$ is 
    \begin{equation}\label{eq:Hdef}
    \vec{H}(\vec{A}, \vec{B}, \bvarphi) = \begin{bmatrix}
        \vec{I}_{r_0} + \vec{A}^{\tran}\bvarphi^{-1}\vec{A} & \vec{A}^{\tran}\bvarphi^{-1}\vec{B} \\
        \vec{B}^{\tran}\bvarphi^{-1}\vec{A}^{\tran} & \vec{I}_{r} + \vec{B}^{\tran}\bvarphi^{-1}\vec{B}
    \end{bmatrix},
    \end{equation}
    and similarly $\vec{H}_c$ is
    \[
    \vec{H}_c(\vec{A}_c, \vec{B}_c, \bvarphi_c) = \vec{I}_{r_0+r_c} + [\vec{A}_c\ \vec{B}_c]^{\tran}\bvarphi_c^{-1}[\vec{A}_c\ \vec{B}_c].
    \]
    Further, the block-diagonal structure of $\vec{B}$ yields that the lower-right part of $\vec{H}$ is block-diagonal with blocks of size $r_c \times r_c, c=1,\dots, C$. The lower-right part of $\vec{H}_{c}$ equals the $c$th block of the lower-right part of $\vec{H}$. Combining the bounds
    \begin{equation*}
    \begin{aligned}
     \rank(\tilde{\vec{A}}\tilde{\vec{A}}^{\tran} + \tilde{\vec{B}}\tilde{\vec{B}}^{\tran}) &\leq r_0 + r,\\
     \rank(\tilde{\vec{A}}_c\tilde{\vec{A}}_c^{\tran} + \tilde{\vec{B}}_c\tilde{\vec{B}}_c^{\tran}) &\leq \min(n_c, r_0 + r_c),
     \end{aligned}
    \end{equation*}
    with \eqref{eq:rankeqns} yields that 
    \begin{equation}\label{eq:HrankUB}
    \begin{aligned}
    \rank\left(\vec{H}(\vec{A}, \vec{B}, \bvarphi)\right) &\leq 2(r_0 +r) - n, \\
    \rank(\vec{I}_{r_c} + \vec{B}^{\tran}_{c}\bvarphi_{c}^{-1} \vec{B}_c) &\leq \min\{r_c, 2(r_0 + r_c) - n_c\}, \\
    \end{aligned}
    \end{equation}
    for $c=1,\dots, C$. {To establish the second line of \eqref{eq:HrankUB}, we combine \eqref{eq:rankeqns} with the above bounds to yield that $\rank(\vec{H}_c) \leq \min\{r_0 + r_c, 2(r_0 + r_c) - n_c\}$. As $\vec{I}_{r_c} + \vec{B}_c^{\tran}\bvarphi_{c}^{-1}\vec{B}_c$ is the lower-right block of $\vec{H}_{c}$ of size $r_c \times r_c$, its rank is bounded above by the minimum of the block size and rank of the whole matrix as $\min\{r_c, \min\{r_0 +r_c, 2(r_0 + r_c) - n_c)\}$. This expression then equals the RHS of the second line of \eqref{eq:HrankUB}.}
    
   Therefore, showing that the set of $[\vec{A}\ \vec{B}]$ where $\vec{H}(\vec{A}, \vec{B}, \bvarphi)$ satisfies \eqref{eq:HrankUB} for some invertible $\bvarphi$ is null  implies that $\mathcal{U}^*_{\emptyset}$ is null as well. If either $2(r_0+r)- n < 0$ or there is a $c$ with $2(r_0 + r_c) - n_c < 0$, a bound in \eqref{eq:HrankUB} is negative and so $\mathcal{U}_{\emptyset}$ is empty. For the remainder of this case, assume that $2(r_0 + r)\geq n$ and $2(r_0 + r_c) \geq n_c$ for all $c$.

     For the codomain of $\vec{H}$, let $\mathcal{S} \subset \mathrm{Sym}(r_0 + r)$ contain the vector space of all symmetric matrices whose $r\times r$ lower-right part is block-diagonal with blocks of sizes $r_c \times r_c$ for $c=1,\dots, C$, which has dimension
     \[
     \dim\mathcal{S} = \textstyle \frac{r_0(r_0+1)}{2} + r_0r + \textstyle \sum_{c=1}^C \frac{r_c (r_c+1)}{2}.
     \]
     Recall that both $(\mathbb{A} \oplus \mathbb{B})^{**}$, which contains $[\vec{A}\ \vec{B}]$ satisfying the maximal rank submatrix condition of Lemma \ref{lem:maxranksubload}, and the non-singular diagonal matrices $\mathrm{Diag}^*(n)$ are open submanifolds of their respective vector spaces. So, $\vec{H}$ is a smooth map from the product manifold $(\mathbb{A} \oplus \mathbb{B})^{**} \times \mathrm{Diag}^*(n)$ to $\mathcal{S}$. The differential $d\vec{H}$ of this map at $(\vec{A}, \vec{B}, \bvarphi)$ is
     \begin{equation*}
     \begin{aligned}
     d\vec{H} &= [\vec{A}\ \vec{B}]^{\tran}\bvarphi^{-1}[d\vec{A}\ d\vec{B}] + [d\vec{A}\ d\vec{B}]^{\tran}\bvarphi^{-1}[\vec{A}\ \vec{B}]\\ 
     &+ [\vec{A}\ \vec{B}]^{\tran}\bvarphi^{-1}d\bvarphi\bvarphi^{-1}[\vec{A}\ \vec{B}].
     \end{aligned}
     \end{equation*}
     The differential is surjective. To see this, consider the subspace of  tangent vectors $(d\vec{A}, d\vec{B}, \vec{0})$ with 
     $
     [d\vec{A}\ d\vec{B}] = \bvarphi^{-1}[\vec{A}\ \vec{B}]\vec{L},
     $
     for $\vec{L} \in \mathbb{R}^{(r_0 + r) \times (r_0 + r)}$ lower-triangular and partitioned as $\vec{L} = [\vec{L}_{11}\ \vec{0}; \vec{L}_{12}\ \vec{L}_{22}]$,
     where $\vec{L}_{22}$ is block-diagonal with $c$th block of size $r_c \times r_c$. The space of such $\vec{L}$ has equal dimension to $\mathcal{S}$. It can be verified that $\bvarphi^{-1}[\vec{A}\ \vec{B}]\vec{L}$ has structural zeros in the appropriate places, and so is a valid choice for $[d\vec{A}\ d\vec{B}]$. On this subspace, $d\vec{H}$ is injective. Suppose that, for some $\vec{L}$ with the above structure, 
     \[
     d\vec{H} = \vec{0} = [\vec{A}\ \vec{B}]^{\tran}\bvarphi^{-2}[\vec{A}\ \vec{B}]\vec{L} + \vec{L}^{\tran}[\vec{A}\ \vec{B}]^{\tran}\bvarphi^{-2}[\vec{A}\ \vec{B}].
     \]
     The matrix  $[\vec{A}\ \vec{B}]^{\tran}\bvarphi^{-2}[\vec{A}\ \vec{B}]$ is positive-definite as $\bvarphi^{-2}$ is positive and $[\vec{A}\ \vec{B}]$ is FCR. So, $[\vec{A}\ \vec{B}]^{\tran}\bvarphi^{-2}[\vec{A}\ \vec{B}]$ has Cholesky-type decomposition $\vec{U}\vec{U}^{\tran}$ with $\vec{U}$ upper-triangular and non-singular. This is obtained by taking the usual Cholesky decomposition of the original matrix with the order of rows and columns reversed. So, the above equation can be manipulated to yield $\vec{U}^{\tran}\vec{L}(\vec{U}^{-1})^{\tran} = -\vec{U}^{-1}\vec{L}^{\tran}\vec{U}$ and so $\vec{U}^{\tran}\vec{L}(\vec{U}^{-1})^{\tran}$ is skew-symmetric. However, as the LHS is lower-triangular while the right is upper-triangular, we must also have that $\vec{U}^{\tran}\vec{L}(\vec{U}^{-1})^{\tran}$ diagonal. Diagonal skew-symmetric matrices must be zero, and $\vec{U}^{\tran}\vec{L}(\vec{U}^{-1})^{\tran}$ implies $\vec{L}=\vec{0}$ as $\vec{U}$ is invertible. So, $d\vec{H}$ is injective on a subspace with the same dimension as the codomain, and so $d\vec{H}$ is surjective. As $(\vec{A}, \vec{B}, \bvarphi) \in (\mathbb{A} \oplus \mathbb{B})^{**} \times \mathrm{Diag}^*(n) $ was arbitrary, $\vec{H}$ is a smooth submersion.

     Next, we will show that the subset of $\mathcal{S}$ where the rank conditions  \eqref{eq:HrankUB} are satisfied can be written as a union of embedded submanifolds. For any index set $\alpha \subset \{1,\dots,(r_0+r)\}$, relate $\alpha$ to the block-structure of $\vec{H}$ by defining $\rho_0\equiv |\alpha \cap \{1,\dots,r_0\}|$ and $\rho_c \equiv |\alpha \cap \{(r_{<c}+1),\dots,(r_{<c} + r_c)\}|$ for each $c$. Then for those $\alpha$ where $\brho \equiv [\rho_0, \rho_1, \dots, \rho_C]$ satisfies 
     \begin{equation}\label{eq:rhoconds}
     \begin{aligned}
     \rho_0 + \textstyle \sum_{c=1}^C \rho_c \leq 2(r_0 + r) - n, \\
     \rho_c \leq 2(r_0 + r_c) - n_c,
     \end{aligned}
     \end{equation}
     let $\mathcal{S}_{\alpha}$ be the subset of $\vec{S} \in \mathcal{S}$ where the principal submatrix $\vec{S}[\alpha, \alpha]$ is non-singular {and $\rank(\vec{S}) = |\alpha|$}. If $\vec{H}(\vec{A}, \vec{B}, \bvarphi)$ satisfies \eqref{eq:HrankUB}, then $\vec{H}(\vec{A}, \vec{B}, \bvarphi) \in \mathcal{S}_{\alpha}$ for some $\alpha$ satisfying \eqref{eq:rhoconds}. For $\vec{S} \in \mathcal{S}_{\alpha}$, symmetry and invertibility of $\vec{S}[\alpha, \alpha]$ implies the complementary submatrix, $\vec{S}[\alpha^{c}, \alpha^{c}]$, is a smooth function of $\vec{S}[\alpha, \alpha]$ and $\vec{S}[\alpha, \alpha^c]$. This fact is equivalent to the well-known matrix completion result that, for $\vec{C} \in \Sym(n)$ partitioned as $\vec{C} = [\vec{C}_{11}\ \vec{C}_{12};\ \vec{C}^{\tran}_{12}\ \vec{C}_{22}]$ with $\vec{C}_{11}$ invertible and $\rank(\vec{C}_{11}) = \rank(\vec{C})$, we have $\vec{C}_{22} = \vec{C}_{12}\vec{C}^{-1}_{11}\vec{C}^{\tran}_{12}$.

    For all $\vec{S} \in \mathcal{S}_{\alpha}$, the submatrices $\vec{S}[\alpha, \alpha]$ and $\vec{S}[\alpha, \alpha^c]$ inherit structural zeros from $\mathcal{S}$ which are determined by $\alpha$. Define the vector space $\mathcal{W}_\alpha \subset \mathrm{Sym}(\rho_0 + \rho)$ containing all symmetric matrices with structural zeros in locations which match those of  $\vec{S}[\alpha, \alpha]$, and similarly let $\mathcal{Y}_{\alpha} \subset \mathbb{R}^{\rho \times (r_0+r - \rho)}$ be the vector space containing all matrices with structural zeros matching those of $\vec{S}[\alpha, \alpha^c]$. As $\vec{I}_{|\alpha|}$ is in  $\mathcal{W}_\alpha$, the subset $\mathcal{W}^*_\alpha$ containing only non-singular matrices is the preimage of $\det_{\alpha}^{-1}(\mathbb{R} \setminus \{0\})$ and so is a non-empty open submanifold of the same dimension as $\mathcal{W}_\alpha$. Here $\det_{\alpha} : \mathcal{W}_{\alpha} \to \mathbb{R}$ is the  usual determinant with domain restricted to $\mathcal{W}_{\alpha}$.  The dimensions of these spaces are
    \begin{equation}\label{eq:dimHparts}
    \begin{aligned}
    \dim(\mathcal{W}_{\alpha}) &= \textstyle \frac{\rho_0(\rho_0+1)}{2} + \rho_0\rho + \textstyle\sum_{c=1}^C \frac{\rho_c(\rho_c+1)}{2}, \\
    \dim(\mathcal{Y}_{\alpha})  &=  \rho_0(r_0 + r - \rho_0) - \rho_0\rho +\textstyle\sum_{c=1}^C \rho_c(r_c - \rho_c).
    \end{aligned}
    \end{equation}
    Then there is the obvious embedding from the product manifold $\mathcal{W}^*_{\alpha} \times \mathcal{Y}_{\alpha}$ into $\mathcal{S}$ obtained by setting $\vec{S}[\alpha, \alpha]$ equal to the first component, $\vec{S}[\alpha, \alpha^c]$ and $\vec{S}[\alpha^c, \alpha]$ to the second component and its transpose respectively, then smoothly obtaining $\vec{S}[\alpha^c, \alpha^c]$ as the unique low-rank matrix completion. So, $\mathcal{S}_{\alpha}$ can be treated as an embedded submanifold of dimension $\dim(\mathcal{W}_{\alpha})  + \dim(\mathcal{Y}_{\alpha})$.

    As $\mathcal{S}_{\alpha}$ is an embedded submanifold, $\vec{H}$ is automatically transverse for $\mathcal{S}_{\alpha}$ by virtue of being a submersion. So, a standard application of Sard's Theorem (see, e.g., \cite[Thm. 6.30]{Lee2002}) ensures that $\vec{H}^{-1}(\mathcal{S}_{\alpha})$ is an embedded submanifold of $(\mathbb{A} \oplus \mathbb{B})^{**} \times \mathrm{Diag}^*(n)$ with codimension equal to the codimension of $\mathcal{S}_{\alpha}$ in $\mathcal{S}$, namely
    \begin{equation*}\label{eq:Scodim}
    \begin{aligned}
    \mathop{\mathrm{codim}}\mathcal{S}_{\alpha} &= \textstyle\frac{r_0(r_0+1) -\rho_0(\rho_0 +1)}{2}  + \textstyle\sum_{c=1}^{C} \frac{r_c(r_c+1) - \rho_c(\rho_c+1)}{2} \\
    &+ (r_0 - \rho_0)r - \rho_0(r_0 - \rho_0) - \textstyle\sum_{c=1}^{C} \rho_c(r_c - \rho_c).
    \end{aligned}
    \end{equation*}
    Let $\pi$ be the projection map from $(\mathbb{A} \oplus \mathbb{B})^{**} \times \mathrm{Diag}^*(n)$ to $(\mathbb{A} \oplus \mathbb{B})^{**}$. Dimensional considerations \cite[p 131]{Lee2002} then imply that  $\pi\left (\vec{H}^{-1}(\mathcal{S}_{\alpha}\right))$ is null if
    \begin{equation}\label{eq:dimcrit}
    \dim(\mathbb{A} \oplus \mathbb{B})^* + n - \mathop{\mathrm{codim}}\mathcal{S}_{\alpha} <  \dim (\mathbb{A} \oplus \mathbb{B})^*.
    \end{equation}
    If the above inequality is satisfied for all $\alpha$ with $\brho$ meeting \eqref{eq:rhoconds}, then $\mathcal{U}^*_{\beta}$ is null for $\beta = \emptyset$. This follows as for any $[\vec{A}\ \vec{B}]$ in $\mathcal{U}_{\emptyset}$, $\vec{H}(\vec{A}, \vec{B},\bvarphi)$ is in $\mathcal{S}_{\alpha}$ for some $\bvarphi$ and some $\alpha$ satisfying \eqref{eq:rhoconds}. Hence $\mathcal{U}^*_{\emptyset}$ is a subset of $\bigcup_{\alpha} \pi\left (\vec{H}^{-1}(\mathcal{S}_{\alpha}\right))$, and the latter is a finite union of null sets. 
    

    \subsubsection*{Case $2$: $\bvarphi$ singular}
    We will show that $\mathcal{U}_{\beta}$ with $|\beta| > 0$ is also null by reducing to the non-singular case with smaller channel sizes and factor numbers. To do so, let $j_c \equiv |\beta \cap \{(r_{<c}+1),\dots,(r_{<c}+r_c)\}|$ be the number of zeros in the $c$th channel on the diagonal of $\bvarphi$ and let the channels be numbered such that $j_1 \geq j_2 \geq \dots \geq j_C$.
    
    For the first channel, we can take $\bvarphi = \mathrm{Diag}(\vec{0}_{j_1}, \bvarphi')$ without loss of generality by permuting $\vec{x}_1$. Continuing to let $\vec{A}_1$ and $\vec{B}_1$ be the common and {distinct} factor loadings for channel $1$ respectively, define the submatrices $\vec{A}_{11}$ and $\vec{B}_{11}$ containing the first $j_1$ rows of $\vec{A}_1$ and $\vec{B}_1$ respectively. Similarly define $\tilde{\vec{A}}_{11}$ and $\tilde{\vec{B}}_{11}$ with respect to $\tilde{\vec{A}}_1$ and $\tilde{\vec{B}}_1$. The remaining submatrices $\vec{A}_{12}, \vec{B}_{12}$ and $\tilde{\vec{A}}_{12}, \tilde{\vec{B}}_{12}$ contain the last $n_1 - j_1$ rows of $\vec{A}_1, \vec{B}_1$ and $\tilde{\vec{A}}_1, \tilde{\vec{B}}_1$ respectively. Finally, let $\vec{A}' = [\vec{A}_{12}^{\tran}\ \vec{A}_{2}^{\tran}\ \dots\ \vec{A}_{C}]^{\tran}$ and $\vec{B}' = \blkdiag(\vec{B}_{12}, \vec{B}_{2}, \dots, \vec{B}_{C})$ be the loadings after exclusion of the top $j_1$ rows, with $\tilde{\vec{A}}'$ and $\tilde{\vec{B}}'$ being similar.

    With these definitions, the zeros of $\bvarphi$ show that \eqref{eq:nonseplowrank} implies 
    \begin{equation}\label{eq:ULblock}
    \vec{A}_{11}\vec{A}_{11}^{\tran} + \vec{B}_{11}\vec{B}_{11}^{\tran} = \tilde{\vec{A}}_{11}\tilde{\vec{A}}_{11}^{\tran} + \tilde{\vec{B}}_{11}\tilde{\vec{B}}_{11}^{\tran}.
    \end{equation}
    As $\bvarphi$ is diagonal, the off-diagonal blocks are also equal, 
        \vspace*{-0.4em}
    \begin{equation}\label{eq:URblock}
    [\vec{A}_{11}\ \vec{B}_{11}\ \vec{0}_{r_{>1}}] [\vec{A}'\ \vec{B}']^{\tran} = [\tilde{\vec{A}}_{11}\ \tilde{\vec{B}}_{11}\ \vec{0}_{r_{>1}}] [\tilde{\vec{A}}'\ \tilde{\vec{B}}']^{\tran}.
    \end{equation}
    Next, recall that the generalized Schur complement of $\vec{A}_{11}\vec{A}_{11}^{\tran} + \vec{B}_{11}\vec{B}_{11}^{\tran}$ in $\vec{A}\vec{A}^{\tran} + \vec{B}\vec{B}^{\tran}$ is
    \vspace*{-0.4em}
    \begin{equation}\label{eq:gsc}
    \vec{A}'\vec{A}'+\vec{B}'\vec{B}' - [\vec{A}'\ \vec{B}'] \vec{W} [\vec{A}'\ \vec{B}']^{\tran},
    \end{equation}
    where $\vec{W}$ is defined as 
    \vspace*{-0.4em}
    \[
    \vec{W} = [\vec{A}_{11} \vec{B}_{11}\ \vec{0}]^{\tran} (\vec{A}_{11}\vec{A}_{11}^{\tran} + \vec{B}_{11}\vec{B}_{11}^{\tran})^{-}[\vec{A}_{11} \vec{B}_{11}\ \vec{0}],
    \]
    with $(\cdot)^{-}$ being the Moore-Penrose psuedo-inverse. As $\vec{A}\vec{A}^{\tran} + \vec{B}\vec{B}^{\tran}$ is positive semi-definite, the generalized Schur complement is uniquely defined \cite[Ch. 6]{zhang2005}, and so \eqref{eq:gsc} equals 
    \begin{equation}\label{eq:gscprod}
        [\vec{A}'\ \vec{B}'] \begin{bmatrix}
            \vec{P} & \vec{0}\\
            \vec{0} & \vec{I}_{r_{>1}}
        \end{bmatrix}
        \begin{bmatrix}
            \vec{A}'^{\tran} \\
            \vec{B}'^{\tran}
        \end{bmatrix},
    \end{equation}
    where $\vec{P}$ is the orthogonal projection onto $\kernel([\vec{A}_{11}\ \vec{B}_{11}])$, which has dimension $(r_0 + r_1 - j_1)_{+}$ by maximal rank submatrix condition. To represent $\vec{P}$, note that $\dim \kernel(\vec{B}_{11})$ is $r'_1 \equiv (r_1 - j_c)_{+}$ as $\vec{B}_{11}$ is a $j_1 \times r_1$ submatrix of $[\vec{A}_1\ \vec{B}_1]$. Choosing $\vec{v}'_{1}, \dots, \vec{v}'_{r'_1}$ as an orthogonal basis for $\kernel(\vec{B}_{11})$, the vectors $\vec{v}_i = [\vec{0}^{\tran}_{r_0}\ \vec{v}_{i}]^{\tran}$ are also in $\kernel([\vec{A}_{11}\ \vec{B}_{11}])$. To the list $\vec{v}_1, \dots, \vec{v}_{r'_1}$, we can extend to an orthogonal basis of $\kernel([\vec{A}_{11}\ \vec{B}_{11}])$ by adding $\vec{w}_{1}, \dots, \vec{w}_{r'_0}$ where 
    \vspace*{-0.4em}
    \[
    r'_0 \equiv (r_0 + r_1 - j_1)_{+} - (r_1 - j_1)_+ = [r_0 - (j_1 - r_1)_+]_{+}.
    \]
    If $\vec{W} = [\vec{w}_1\ \dots\ \vec{w}_{r'_0}]$ and $\vec{V}' = [\vec{v}'_1\ \dots\ \vec{v}_{r'_1}]$, then if
    \vspace*{-0.7em}
    \begin{equation*}
    \vec{D} \equiv  \begin{bNiceMatrix}[first-row,last-col=5]
        r'_{0} & r_0 - r'_0 & r'_1 & r_1 - r'_1 & \\
        \vec{W}_{0} &\vec{0} & \vec{0} & \vec{0} & r_0 \\
        \vec{W}_{1} &\vec{0} & \vec{V}' & \vec{0}& r_1
    \end{bNiceMatrix}\quad ,
    \end{equation*}
    the projection $\vec{P}$ equals $\vec{D}\vec{D}^{\tran}$, where $\vec{W}_0$ and $\vec{W}_1$ contain the first $r_0$ rows and the remaining $r_1$ rows of $\vec{W}$ respectively.
    
    So, the inner term in \eqref{eq:gscprod} factors into $\vec{R}\vec{R}^\tran$ where $\vec{R} = [\vec{D}\ \vec{0};\ \vec{0}\ \vec{I}_{r_{>1}}]$.  As $\vec{R}$ has a block lower-triangular form, it can be verified that $[\vec{A}'\ \vec{B}']\vec{R}$ continues to have the appropriate channel structure. The generalized Schur complement can then be represented as $\vec{A}_o\vec{A}_{o}^{\tran} + \vec{B}_o\vec{B}_o^{\tran}$ where $[\vec{A}_o\ \vec{B}_o]$ is obtained from $[\vec{A}'\ \vec{B}']\vec{R}$ by dropping the zero columns.

    Using this representation, we can take the lower $(n-j_1) \times (n-j_1)$ block of \eqref{eq:nonseplowrank} and subtract $[\vec{A}'\ \vec{B}']\vec{W}[\vec{A}'\ \vec{B}']$ from both sides. By the equalities \eqref{eq:ULblock} and \eqref{eq:URblock}, this implies the relation between the generalized Schur complements,
    \[
    \vec{K} \setminus \vec{K}_{11} = \tilde{\vec{K}} \setminus \tilde{\vec{K}}_{11} + \bvarphi' 
    \]
    where $\vec{K} \equiv \vec{A}\vec{A}^{\tran}+\vec{B}\vec{B}^{\tran}$ and $\vec{K}_{11} \equiv \vec{A}_{11}\vec{A}_{11}^{\tran} + \vec{B}_{11}\vec{B}_{11}^{\tran}$, and $\tilde{\vec{K}}, \tilde{\vec{K}}_{11}$ are defined similarly using $\tilde{\vec{A}}, \tilde{\vec{B}}$. As discussed above, this implies there are $\vec{A}^o, \vec{B}^o$ and $\tilde{\vec{A}}^o, \tilde{\vec{B}}^o$ such that
    \[
    \vec{A}^{o}\vec{A}^{o\tran} + \vec{B}^{o}\vec{B}^{o\tran} = \tilde{\vec{A}}\tilde{\vec{A}}^{o\tran} + \tilde{\vec{B}}^{o}\tilde{\vec{B}}^{o\tran} + \bvarphi'.
    \]
     In the case with $j_c = 0$ for $c>1$, The above procedure exhibits a reduction from $\mathcal{U}_{\beta}^*$ into $\mathcal{U}'^*_{\emptyset}$ where $\mathcal{U}'$ is the set of loadings satisfying \eqref{eq:nonseplowrank} for channel sizes $n_1 - j_1, n_2, \dots, n_C$ and factor numbers $r_0', r_1', r_2, \dots, r_C$. As the orthogonal projection $\vec{P}$ varies smoothly with $[\vec{A}_{11}\ \vec{B}_{11}]$, the matrix $\vec{D}$ can be chosen to smoothly vary in $[\vec{A}_{11}\ \vec{B}_{11}]$ and so the reduction is smooth.

    In other cases, the above procedure can be iterated to remove $j_c$ zeros from $\bvarphi_c$ each time, yielding a smooth reduction from $\mathcal{U}_{\beta}^*$ to $\mathcal{U}'^*_{\emptyset}$. Showing $\mathcal{U}^*_{\beta}$ to be null reduces to showing $\mathcal{U}'^*_{\emptyset}$ to be null, where $\mathcal{U}'$ is the non-separable set \eqref{eq:nonseplowrank} with channel sizes $n_1 -j_1, \dots, n_C -j_C$ and factor sizes 
    $r_0' = (r_0 - \textstyle \sum_{c=1}^C (j_c - r_c)_{+})_{+}$, $r'_c = (r_c - j_c)_{+}$ for $c=1,\dots, C$.   

    \subsubsection*{Verification of Condition \ref{cond:phisep}}
    To demonstrate that Condition \ref{cond:phisep} is sufficient to imply that $\mathcal{U}$ is null, first assume that $M$ as defined in \eqref{eq:phisep2} is non-empty. For $\mathcal{U}^*_{\emptyset}$, if the subset of $M$ containing $(\vec{n}', \vec{r}', \brho)$ with  $\vec{n}'=\vec{n}, \vec{r}'=\vec{r}$ is empty, then \eqref{eq:rhoconds} cannot be satisfied with $\brho$ non-negative and so either $2(r_0+r) - n < 0$ or there is a $c$ with $2(r_0 + r_c) - n_c < 0$. That is, $\vec{H}(\vec{A}, \vec{B}, \bvarphi)$ cannot satisfy \eqref{eq:HrankUB} for any $(\vec{A}, \vec{B})$ and so $\mathcal{U}^*_{\emptyset}$ is empty. If the aforementioned subset of $M$ is non-empty, then at least some non-negative $\brho$ satisfying \eqref{eq:rhoconds} is possible. Then the dimension condition \eqref{eq:dimcrit} is equivalent to $\psi(\vec{n}, \vec{r}, \brho) > 0$ with $\psi$ defined in \eqref{eq:MFAidentcritfunc}. As a function of $\rho_0$ alone, $\psi$ is decreasing on $[0, r_0]$, so $\psi$ being positive  when $\rho_0$ is at its maximum feasible value implies $\psi$ is positive for all smaller values of $\rho_0$. So, $\min_{\rho} \psi(\vec{n}, \vec{r}, \rho) > 0$ ensures that \eqref{eq:dimcrit} is satisfied for all valid $\alpha$. Therefore, $\mathcal{U}^*_{\emptyset}$ is null. 
    
    For any index set $\beta$ with $j_c, c=1,\dots, C$ zeros in the $c$th channel on the diagonal of $\bvarphi$, taking $n'_{c} = n_{c} - j_c$ and $r_0', r'_1, \dots r'_C$ as in \eqref{eq:phisep2} effects the reduction of $\mathcal{U}^*_{\beta}$ to $\mathcal{U}'^*_{\emptyset}$ as discussed in the previous subsection. The above argument can be applied to $\mathcal{U}'^*_{\emptyset}$, showing that $\psi^* > 0$ implies $\mathcal{U}'^*_{\emptyset}$ is null for all possible reductions. Therefore, $\mathcal{U}^*_{\beta}$ is  null for all $\beta$. The remaining part of $\mathcal{U}$ is a subset of $(\mathbb{A} \oplus \mathbb{B}) \setminus (\mathbb{A} \oplus \mathbb{B})^{**}$, which is null by Lemma \ref{lem:maxranksubload}. So, $\mathcal{U}$ is a subset of the finite union of the null sets $\mathcal{U}^*_{\emptyset}, \mathcal{U}^*_{\beta}$ for all $\beta$, and $(\mathbb{A} \oplus \mathbb{B}) \setminus (\mathbb{A} \oplus \mathbb{B})^{**}$.

    Finally, if $M$ is empty, then no feasible value for $\brho$ exists for any possible reduction and so $\mathcal{U}^*$ is empty. Then $\mathcal{U} \subset \mathbb{A} \oplus \mathbb{B} \setminus (\mathbb{A} \oplus \mathbb{B})^{**}$ and so is null. As $\mathbb{A} \oplus \mathbb{B}$ is isomorphic to $\mathbb{A} \times \mathbb{B}$, the result of Theorem \ref{thm:phi_separation} follows.
     \end{proof}

    \begin{proof}[Proof of Theorem \ref{prop:local_identifiability}]\label{sec:prop1proof}
    First, note that $d\vec{R}_{\vec{x}\vec{x}}(\veta)$ is equivalent to $d\vec{R}_{\vec{x}\vec{x}}(\vec{A}, \vec{B}, \bPhi)$ when $(\vec{A}, \vec{B}, \bPhi) \in \mathbb{A}^*_{L} \times \mathbb{B}_{L}^* \times \mathrm{Diag}_{\geq 0}(n)$ is obtained from $\veta$ by inverting \eqref{eq:vectorization}. For $(d\vec{A}, d\vec{B}) \in \mathbb{A}_L \times \mathbb{B}_L$ and $d\bPhi \in \mathrm{Diag}(n)$, the differential 
    \begin{equation}
    \label{eq:Rdifferential}
    d\vec{R}_{\vec{x}\vec{x}}= \vec{A}d\vec{A}^{\tran} + d\vec{A}\vec{A}^{\tran} + \vec{B}d\vec{B}^{\tran} + d\vec{B}\vec{B}^{\tran} + d\bPhi 
    \end{equation}
    is a linear map in $(d\vec{A}, d\vec{B}, d\bPhi)$ from a vector space of dimension $L$ as defined in \eqref{eq:unknown_count} to a space of dimension $n(n+1)/2$. Necessary conditions for injectivity can be obtained by dimensionality considerations. First, the dimension of the domain must be no greater than the codomain for the map to be injective, so $L \leq n(n+1)/2$.  The previous inequality is equivalent to  the condition
    \begin{equation*}
    r_0 \leq \frac{1}{2} \left(2n + 1 - \sqrt{8(n+D) + 1}\right),
    \end{equation*}
    where $D$ is defined in \eqref{eq:K1def}.
    Similarly, setting $d\vec{A}, d\bPhi$ and $d\vec{B}_2, \dots, d\vec{B}_C$ to zero (of the appropriate dimensions), the restricted $d\vec{R}_{\vec{x}\vec{x}}$ is a linear map from a vector space of dimension $n_1r_1 - \frac{r_1(r_1-1)}{2}$ into the subspace of symmetric matrices with only the top $n_1 \times n_1$ block being non-zero. Again, the dimension of the restricted domain must be no greater than the codomain, meaning
     $n_1r_1 -  r_1(r_1-1)/2 \leq n_1(n_1+1)/2$.
    This is equivalent to the condition $
    r_1 \leq \frac{1}{2} \left(2n_1 + 1 - \sqrt{8n_1 + 1}\right).
    $
    The same considerations for other blocks yield the conditions that, for all $c=1,\dots, C$, $r_c \leq \frac{1}{2} \left(2n_c + 1 - \sqrt{8n_c + 1}\right)$.
    Combined, this is Condition \ref{cond:fac_count} and so Proposition \ref{prop:local_ident_nesc} is proven.
    
    Next, we will show that the Condition \ref{cond:fac_count} combined with the separability Condition \ref{cond:rotdet} is sufficient for $d\vec{R}_{\vec{x}\vec{x}}$ to be generically injective in $(d\vec{A}, d\vec{B}, d\bPhi)$. Assume that the combined matrix $[\vec{A}\ \vec{B}]$ is FCR, which excludes a null subset of $\mathbb{A}^*_L \times \mathbb{B}_{L}^*$. The proof proceeds in three steps: first showing that $d\vec{A} \mapsto \vec{A}d\vec{A}^\tran + d\vec{A}\vec{A}^\tran$ and $d\vec{B} \mapsto \vec{B}d\vec{B}^\tran + d\vec{B}\vec{B}^\tran$ are separately injective, then showing the sum of the two is generically injective, then finally showing that the map $(d\vec{A}, d\vec{B}, d\bPhi) \mapsto d\vec{R}_{\vec{x}\vec{x}}$ is generically injective and therefore $d\vec{R}_{\vec{x}\vec{x}}(\veta)$ is generically injective.  
    
    Define the linear maps $\vec{F}_{\vec{A}} : \mathbb{A}_L \to \mathrm{Sym}(n)$ and $\vec{F}_{\vec{B}} : \mathbb{B}_L \to \mathrm{Sym}(n)$ as
    $\vec{F}_{\vec{A}}(\vec{X}) \equiv \vec{A}\vec{X}^{\tran} +  \vec{X}\vec{A}^{\tran}$ and
    $\vec{F}_{\vec{B}}(\vec{Y}) \equiv \vec{B}\vec{Y}^{\tran} +  \vec{Y}\vec{B}^{\tran}$.
    Similarly, $\vec{F}_{\vec{A},\vec{B}} : \mathbb{A}_L \times \mathbb{B}_L \to \mathrm{Sym}(n)$ is
    \begin{equation}
    \begin{aligned}
    \label{eq:FABdef}
    \vec{F}_{\vec{A}, \vec{B}}(\vec{X}, \vec{Y}) &\equiv \vec{F}_{\vec{A}}(\vec{X}) + \vec{F}_{\vec{B}}(\vec{Y}) \\
    &=[\vec{A} \ \vec{B}][\vec{X}\ \vec{Y}]^{\tran} + [\vec{X} \ \vec{Y}][\vec{A}\ \vec{B}]^{\tran}.
    \end{aligned}
    \end{equation}
    
   For $\vec{F}_\vec{A}$, an application of Lemma \ref{LT_lemma} implies that $\vec{F}_\vec{A}$ is injective, as $\vec{F}_{\vec{A}}(\vec{X}) = \vec{0}$ implies $\vec{X} = \vec{0}$ by the structure of $\mathbb{A}_L$ and the FCR assumption. Arguing channel-wise for $\vec{F}_{\vec{B}}$, $\vec{F}_{\vec{B}}$ is injective in a similar fashion.

    Second, define the image subspaces as $\mathcal{A} \equiv \imagesp(\vec{F}_{\vec{A}})$ and $\mathcal{B}  \equiv \imagesp(\vec{F}_{\vec{B}})$. Injectivity of $\vec{F}_\vec{A}$ and $\vec{F}_{\vec{B}}$ implies
    \begin{equation*}
    \begin{aligned}
    \mathrm{dim}(\mathcal{A}) &= \textstyle nr_0 - \frac{r_0(r_0-1)}{2} \\
    \mathrm{dim}(\mathcal{B}) &= \textstyle \sum_{c=1}^C n_cr_c - \frac{r_c(r_c{-}1)}{2}.
    \end{aligned}
    \end{equation*}
    The sum map $\vec{F}_{\vec{A}, \vec{B}}$ will be injective iff $\mathcal{A} \cap \mathcal{B} = \{\vec{0}\}$,  as existence of $(\vec{X}, \vec{Y})$ with $\vec{F}_{\vec{A}, \vec{B}}(\vec{X}, \vec{Y}) = \vec{0}$ implies that $\vec{F}_{\vec{A}}(\vec{X}) = -\vec{F}_{\vec{B}}(\vec{Y})$ and so $\vec{F}_{\vec{B}}(\vec{Y})$ is in $\mathcal{A} \cap \mathcal{B}$. If $\mathcal{A} \cap \mathcal{B} = \{\vec{0}\}$, then injectivity of $\vec{F}_{\vec{A}}$ and $\vec{F}_{\vec{B}}$ separately implies $(\vec{X}, \vec{Y}) = (\vec{0}, \vec{0})$. The converse direction follows as if $\vec{G} \in \mathcal{A} \cap \mathcal{B}$ and $\vec{G} \neq \vec{0}$, then exist non-zero $\vec{X}$ and $\vec{Y}$ such that $\vec{F}_{\vec{A}}(\vec{X}) = \vec{G} = \vec{F}_{\vec{B}}(\vec{Y})$ and so $\vec{F}_{\vec{A}, \vec{B}}(\vec{X}, -\vec{Y}) = \vec{0}$.
    
    To show $\mathcal{A} \cap \mathcal{B} = \{\vec{0}\}$, suppose $\vec{F}_{\vec{A},\vec{B}}(\vec{X}, \vec{Y}) = \vec{0}$ and so
    \begin{equation}
    \label{eq:injective2}
    \vec{A}\vec{X}^{\tran} + \vec{X}\vec{A}^{\tran} = -\vec{B}{\vec{Y}^{\tran}} - \vec{Y}\vec{B}^{\tran}. 
    \end{equation}
    Next, suppose there exists some $\vec{v} \in \mathbb{R}^{n}$ with $\vec{v} \in \kernel(\vec{A}^{\tran})\cap\kernel(\vec{B}^{\tran})$ but at least one of $\vec{X}^{\tran}\vec{v}$ or $\vec{Y}^{\tran}\vec{v}$ is non-zero. Then applying the transformations in \eqref{eq:injective2} to $\vec{v}$, we obtain 
    \[
    \vec{A}\vec{X}^{\tran}\vec{v} = \vec{B}(-\vec{Y})^{\tran}\vec{v}.
    \]
    The LHS is then in $\imagesp(\vec{A})$ and the RHS is in $\imagesp(\vec{B})$. As $\imagesp(\vec{A}) \cap \imagesp(\vec{B}) = \{\vec{0}\}$ is implied by the FCR assumption on $[\vec{A}\ \vec{B}]$, we must have that both $\vec{A}\vec{X}^{\tran}\vec{v}$ and $ \vec{B}(-\vec{Y})^{\tran}\vec{v}$ are zero.
    However, $\vec{A}$ and $\vec{B}$ both being full rank implies that $\kernel(\vec{A}) = \{\vec{0}\}$ and $\kernel(\vec{B}) = \{\vec{0}\}$. As at least one of $\vec{X}^{\tran}\vec{v}$ and $(-\vec{Y})^{\tran}\vec{v}$ is non-zero by assumption, at least one of $\vec{A}\vec{X}^{\tran}\vec{v}$ and $\vec{B}(-\vec{Y})^{\tran}\vec{v}$ is non-zero, yielding a contradiction. So, for all $\vec{v} \in \kernel(\vec{A}^{\tran})\cap\kernel(\vec{B}^{\tran})$, $\vec{X}^{\tran}\vec{v}$ and $\vec{Y}^{\tran}\vec{v}$ are zero. 
    
    Hence, we have that $\kernel(\vec{X}^{\tran}) \supset \kernel(\vec{A}^{\tran})\cap\kernel(\vec{B}^{\tran})$ while $\kernel(\vec{Y}^{\tran}) \supset \kernel(\vec{A}^{\tran})\cap\kernel(\vec{B}^{\tran})$, which is equivalent to $\imagesp(\vec{X}) \subset \imagesp(\vec{A}) \oplus \imagesp(\vec{B})$ and $\imagesp(\vec{Y}) \subset \imagesp(\vec{A}) \oplus \imagesp(\vec{B})$. So, all columns of $[\vec{X}\ \vec{Y}]$ can be written as a linear combination of the columns of $[\vec{A}\ \vec{B}]$ and there exists some $\vec{W} \in \mathbb{R}^{(r_0 + r) \times (r_0 + r)}$ such that
    \begin{equation}
    \label{eq:Wdef}
    [\vec{X}\ \vec{Y}] = [\vec{A}\ \vec{B}]\vec{W}.
    \end{equation}
    Combining \eqref{eq:Wdef} and \eqref{eq:FABdef}, $\vec{F}_{\vec{A},\vec{B}}(\vec{X}, \vec{Y}) = \vec{0}$ can be written as
    $
    [\vec{A}\ \vec{B}](\vec{W}^\tran + \vec{W})[\vec{A}\ \vec{B}]^{\tran} = \vec{0}.
    $
    As $[\vec{A}\ \vec{B}]$ is FCR, left and right multiplying by $([\vec{A}\ \vec{B}]^{\tran}[\vec{A}\ \vec{B}])^{-1}[\vec{A}\ \vec{B}]^{\tran}$ and its transpose yields that $(\vec{W} + \vec{W}^{\tran}) = \vec{0}$, so $\vec{W}$ is skew-symmetric.

    However, Proposition \ref{prop:genericrotdet} implies that $\vec{W}$ must be zero unless $(\vec{A}, \vec{B})$ belong to a null subset of $\mathbb{A}_L^* \times \mathbb{B}^*_L$. To see this, recall that the Cayley transform gives that there exists a  $\vec{Q} \in \mathrm{O}(r_0 +r)$ such that
    $
    \vec{W} = (\vec{Q} - \vec{I})(\vec{Q} + \vec{I})^{-1}.
    $
    So, \eqref{eq:Wdef} is equivalent to 
    \[
    [\vec{X}\ \vec{Y}] + [\vec{A}\ \vec{B}] = \left([\vec{A}\ \vec{B}] -  [\vec{X}\ \vec{Y}]\right)\vec{Q}.
    \]
    The above implies that $(\vec{X}+\vec{A}, \vec{Y}+\vec{B}) \sim_{1} (\vec{A}-\vec{X}, \vec{B}-\vec{Y})$, where both pairs are in $\mathbb{A}_L \times \mathbb{B}_L$. Note that $\vec{F}_{\vec{A},\vec{B}}(\vec{X},\vec{Y}) = \vec{0}$ implies that $\vec{F}_{\vec{A},\vec{B}}(\epsilon\vec{X},\epsilon\vec{Y}) = \vec{0}$ for any $\epsilon$, meaning that $(\vec{A}-\vec{X},\vec{B}-\vec{Y})$ can be assumed to belong to an arbitrary neighborhood of $(\vec{A}, \vec{B})$ in $\mathbb{A}_L^* \times \mathbb{B}_L^*$. By the proof of Proposition \ref{prop:genericrotdet}, the subset $\tilde{\mathcal{C}} \subset \mathbb{A}_L^* \times \mathbb{B}_L^*$ where the full rank submatrix condition of Proposition \ref{prop:submatrices_rotdet} is not satisfied is closed and null. Generically, $(\vec{A}, \vec{B})$ is in the complement $\mathbb{A}_L^* \times \mathbb{B}_L^* \setminus \tilde{\mathcal{C}}$ and so $(\vec{A}-\vec{X},\vec{B}-\vec{Y})$ also belongs to the complement for small enough $(\vec{X}, \vec{Y})$ as $[\vec{A}\ \vec{B}]$ is FCR. So, Proposition \ref{prop:submatrices_rotdet} applies to $(\vec{A}-\vec{X},\vec{B}-\vec{Y})$, meaning that $\vec{Q}$ must be block-diagonal and so $(\vec{X}+\vec{A}, \vec{Y}+\vec{B}) \sim_{2} (\vec{A}-\vec{X}, \vec{B}-\vec{Y})$. By Proposition \ref{prop:LT}, the $\sim_2$-representatives in $\mathbb{A}^*_L \times \mathbb{B}_L^*$ are unique, hence $(\vec{X}, \vec{Y}) = (\vec{0}, \vec{0})$. Thus, $\vec{F}_{\vec{A},\vec{B}}$ is generically injective.
    
    For the third step,  we complete the proof by showing that
    \begin{equation}\label{eq:diagint}
    (\mathcal{A} \oplus \mathcal{B}) \cap \mathrm{Diag}(n) = \{\vec{0}\},
    \end{equation}
    which is equivalent to showing that the  differential \eqref{eq:Rdifferential} is injective. 
    To show \eqref{eq:diagint}, we examine the orthogonal complement $\left((\mathcal{A} \oplus \mathcal{B}) \cap \mathrm{Diag}(n)\right)^{\perp}$. Standard properties of the subspace lattice 
    show that $\left((\mathcal{A} \oplus \mathcal{B}) \cap \mathrm{Diag}(n)\right)^{\perp}$ equals $(\mathcal{A}^\perp \cap \mathcal{B}^{\perp}) \cap (\mathcal{A}^{\perp} \cap \mathcal{B}^{\perp} \cap \mathrm{Diag}(n)^{\perp})^{\perp}  \oplus \mathrm{Diag}(n)^{\perp}$. As $\dim(\mathrm{Diag}(n)^{\perp})$ is $n(n-1)/2$, it suffices to show
    \begin{equation}
    \label{eq:nondiagparts}
    n = \dim\big((\mathcal{A}^\perp \cap \mathcal{B}^{\perp}) \cap (\mathcal{A}^{\perp} \cap \mathcal{B}^{\perp} \cap \mathrm{Diag}(n)^{\perp})^{\perp}\big ). 
    \end{equation}
    Expanding $\mathcal{B}^{\perp}$ using that $\vec{F}_{\vec{B}}(\vec{Y})$ is channel-structured block-diagonal, we see that any matrix with zeros in the block-diagonal is within $\mathcal{B}^{\perp}$. Let $\mathcal{Y}$ be the subspace of such matrices,
     \begin{equation*}
     \mathcal{Y} = \mathrm{Span}\big( \{\vec{e}_{i}\vec{e}_{j}^{\tran} + \vec{e}_{j}\vec{e}_{i}^{\tran} :\ \exists\ c\ \mathrm{s.t.}\ i \leq \textstyle r_{<c}+r_{c} < j \} \big).
     \end{equation*}
      The structure of $\mathcal{Y}$ then implies $\mathcal{Y} \subset \mathcal{B}^{\perp}$ and therefore $\mathcal{B}^{\perp} = \mathcal{Y} \oplus (\mathcal{B}^{\perp} \cap \mathcal{Y}^{\perp})$. The subspace $\mathcal{B}^{\perp} \cap \mathcal{Y}^{\perp}$ contains all channel-structured block-diagonal matrices where for all $c$, the $c$th block is orthogonal to all matrices of the form $\vec{B}_c\vec{Y}^{\tran}_c + \vec{Y}_c\vec{B}^{\tran}_c$. Hence, the subspace $\mathcal{B}^{\perp} \cap \mathcal{Y}^{\perp}$ be can be written as $\mathcal{W}_1 \oplus \dots \oplus \mathcal{W}_C$, where $\mathcal{W}_c \subset \mathrm{Sym}(n)$ contains the matrices in $\mathcal{B}^{\perp} \cap \mathcal{Y}^{\perp}$ whose entries for blocks other than $c$ are all zero.
     
    The intersection $\mathcal{A}^{\perp} \cap \mathcal{B}^{\perp} \cap \mathrm{Diag}(n)^{\perp}$ equals $\mathcal{A}^{\perp} \cap (\mathcal{W}_1 \oplus \dots \oplus \mathcal{W}_c \oplus \mathcal{Y}) \cap  \mathrm{Diag}(n)^{\perp}$, which in turn equals $ \mathcal{A}^{\perp} \cap (((\mathcal{W}_1 \oplus \dots \oplus \mathcal{W}_c)\cap \mathrm{Diag}(n)^{\perp}) \oplus \mathcal{Y})$
     as $\mathcal{Y} \subset \mathrm{Diag}(n)^{\perp}$. The subspace $(\mathcal{W}_1 \oplus \dots \oplus \mathcal{W}_c)\cap \mathrm{Diag}(n)^{\perp}$ consists of all $\vec{C}\in \mathcal{B}^{\perp} \cap \mathcal{Y}^{\perp}$ with zero diagonal. This space equals $\widetilde{\mathcal{W}}_1 \oplus \dots \oplus \widetilde{\mathcal{W}}_C$, where $\widetilde{\mathcal{W}}_c$ is the subspace of all elements of $\mathcal{W}_c$ with zero diagonal. So, $\dim(\mathcal{A}^{\perp} \cap \mathcal{B}^{\perp} \cap \mathrm{Diag}(n)^{\perp})$ is
    \begin{equation*}
    \dim(\mathcal{A}) + \textstyle\sum_{c=1}^C \dim(\widetilde{\mathcal{W}}_c) + \dim(\mathcal{Y})- \dim(\mathcal{A} + \textstyle\sum_{c=1}^C \widetilde{\mathcal{W}}_c + \mathcal{Y}). 
    \end{equation*}
    The dimension of  $\mathcal{W}_c$ can be written as
    \begin{equation*}
    \begin{aligned}
    \dim(\widetilde{\mathcal{W}}_c) &= \dim(\mathcal{W}_c) + \textstyle\frac{n_c(n_c-1)}{2} - (\textstyle\frac{n_c(n_c+1)}{2} -J_c) \\
    &= \dim(\mathcal{W}_c) - (n_c - J_c),
    \end{aligned}
    \end{equation*}
    where $J_c$ is the dimension of the intersection between the subspace of matrices in $\mathrm{Sym}(n_c)$ that are realized as $\vec{B}_c\vec{Y}_c^{\tran} + \vec{Y}_c\vec{B}_c^\tran$ and the subspace of diagonal matrices. This will be zero iff the individual channel is locally identifiable as a single channel factor model, which will generically occur when \eqref{eq:local_identifiability_bound2} is satisfied, as shown in \cite[Thm. 3.2]{shapiro1985}. 
    If the above is satisfied, $\mathcal{B}\cap\mathrm{Diag}(n) = \{\vec{0}\}$. From this, we obtain that 
    \begin{equation}
    \begin{aligned}
    \label{eq:AperpBperp_expansion}
    \mathcal{A}^{\perp} + \mathcal{B}^{\perp} &= (\mathcal{A}^\perp + \mathcal{B}^\perp) \cap (\mathcal{A}^\perp + (\mathcal{B}^{\perp} \cap \mathrm{Diag}(n)^{\perp})) \\&\quad \oplus
    (\mathcal{A}^\perp + \mathcal{B}^\perp) \cap (\mathcal{A}\cap(\mathcal{B} \oplus \mathrm{Diag}(n))). \\
    \end{aligned}
    \end{equation}
    However, the term $\mathcal{A}\cap(\mathcal{B} \oplus \mathrm{Diag}(n))$ can be shown to be $\{ \vec{0}\}$ by the results from \cite{Shapiro1982} for single-channel FA when \eqref{eq:local_identifiability_bound1} is satisfied. In particular, if $\vec{C} \in \mathcal{A}\cap(\mathcal{B} \oplus \mathrm{Diag}(n))$, then as $\vec{C} \in \mathcal{A}$, \cite[Lemma 2.1]{Shapiro1982} implies that
    \begin{equation*}
    \vec{v}_{i}^{\tran} \vec{C} \vec{v}_j = 0,\quad 1\leq i \leq j \leq n-r_0,
    \end{equation*}
    where $\vec{v}_1, \dots, \vec{v}_{n-r_0} \in \mathbb{R}^{n}$ is a basis for $\kernel(\vec{A}^{\tran})$. As the structure of $\mathbb{A}$ does not restrict $\kernel(\vec{A}^{\tran})$, all the above linear constraints on $\vec{C}$ are generically independent. There are $\binom{n-r_0}{2}$ constraints, and $\vec{C} \in (\mathcal{B} \oplus \mathrm{Diag}(n))$ has $\mathrm{dim}(\mathcal{B}\oplus \mathrm{Diag}(n))$ degrees of freedom.  Hence, $\mathcal{A}\cap(\mathcal{B} \oplus \mathrm{Diag}(n))$ is $\{\vec{0}\}$ iff
    \begin{equation*}
    \textstyle\frac{(n-r_0)(n-r_0+1)}{2} \geq  \dim(\mathcal{B}) +  \dim(\mathrm{Diag}(n)),
    \end{equation*}
    as the RHS is $\dim(\mathcal{B} \oplus \mathrm{Diag}(n))$.
    This is equivalent to the condition \eqref{eq:local_identifiability_bound1}, and therefore $\mathcal{A}\cap(\mathcal{B} \oplus \mathrm{Diag}(n)) = \{\vec{0}\}$.
    This implies that the RHS of \eqref{eq:AperpBperp_expansion} equals $(\mathcal{A}^\perp + (\mathcal{B}^{\perp} \cap \mathrm{Diag}(n)^{\perp}))$.
    Finally, \eqref{eq:nondiagparts} can be expanded as,
    \begin{equation*}
    \begin{aligned}
    &\dim\big((\mathcal{A}^\perp \cap \mathcal{B}^{\perp}) \cap (\mathcal{A}^{\perp} \cap \mathcal{B}^{\perp} \cap \mathrm{Diag}(n)^{\perp})^{\perp}\big )\\
     &=\dim(\mathcal{A}^\perp)  + \dim(\mathcal{B}^\perp) - \dim(\mathcal{A}^{\perp} + \mathcal{B}^{\perp})\\
    &\quad  - \dim(\mathcal{A}^{\perp} \cap \mathcal{B}^{\perp} \cap \mathrm{Diag}(n)^{\perp}) \\
    &= \dim(\mathcal{A}^\perp) + \textstyle\sum_{c=1}^C\dim(\mathcal{W}_c) + \dim(\mathcal{Y}) \\
    &\quad- \dim(\mathcal{A}^{\perp} + \mathcal{B}^{\perp}) - \dim(\mathcal{A}^\perp) - \textstyle\sum_{c=1}^C \dim(\widetilde{W}_c) \\ &\quad+\dim(\mathcal{A}^\perp + (\mathcal{B}^{\perp} \cap \mathrm{Diag}(n)^\perp)).
    \end{aligned}
    \end{equation*}
    By the previous results, this simplifies to 
    {\small
    \begin{equation*}
    \begin{aligned}
    n {-} \big[\dim(\mathcal{A}^\perp{+}\mathcal{B}^\perp) {-} \dim(\mathcal{A}^\perp{+}(\mathcal{B}^{\perp} \cap \mathrm{Diag}(n)^\perp)) \big]
    {-} \textstyle\sum_{c=1}^C J_c,
    \end{aligned}
    \end{equation*}
    }
    which equals $n$ iff all $J_c = 0$ and 
    \[
    \dim(\mathcal{A}^\perp + (\mathcal{B}^{\perp} \cap \mathrm{Diag}(n)^\perp)) = \dim(\mathcal{A}^\perp + \mathcal{B}^\perp).
    \] This occurs generically when the criteria in Proposition \ref{prop:local_identifiability} are satisfied, and so the differential $d\vec{R}_{\vec{x}\vec{x}}$ will be injective at almost all $(\vec{A}, \vec{B}, \bPhi) \in \mathbb{A}^*_L \times \mathbb{B}_L^* \times \mathrm{Diag}(n)$. This space is isomorphic to $V$, so $d\vec{R}_{\vec{x}\vec{x}}(\veta)$ is generically injective. 
    \end{proof}

\begin{proof}[Proof of Theorem \ref{thm:consistency}]\label{sec:thm2proof}
    Define the function $\ell_0$ over $\veta \in V$,
    \begin{equation*}
    \begin{aligned}
    \ell_0(\veta) &\equiv E[\ell_{T}(\veta)] = \log \det \vec{R}_{\vec{x}\vec{x}}(\veta) + \tr \vec{R}^{-1}_{\vec{x}\vec{x}}(\veta) \bSigma_{\vec{x}\vec{x}}\\
    &= {2D_{KL}\big (\mathcal{N}(\vec{0}, \bSigma_{\vec{x}\vec{x}}) || \mathcal{N}(\vec{0}, \vec{R}_{\vec{x}\vec{x}}(\veta))\big)+ \log\det \bSigma_{\vec{x}\vec{x}}}.
    \end{aligned}
    \end{equation*}
    {By hypothesis, $\ell_0$ has a unique minimum over $V'$ at $\mathring{\veta}$ as $\mathring{\veta}$ belongs to the globally identified set.  Further, this minimum is well-separated. To see this, first note that the sublevel sets of $\ell_0(\veta)$ are compact by the proof of \cite[Thm. 1]{ramirez2020} with $\vec{S} = \bSigma_{\vec{x}\vec{x}}$ and the fact that the map $\veta \mapsto (\vec{A}, \vec{B}, \bPhi)$ is continuous. Let $m = \ell_0(\mathring{\veta})$, which is finite as $\bSigma_{\vec{x}\vec{x}}$ and $\mathring{\vec{R}}_{\vec{x}\vec{x}}$ are positive definite and so yield a finite KL-divergence. For any $h > 0$ and all $\epsilon > 0$, the closed set $V'$ is partitioned into the three sets $B_{\epsilon} \cap L_{h}$, $B^c_{\epsilon} \cap L_h$ and $L_h^{c}$, where 
    \[
    \begin{aligned}
        B_{\epsilon}  &= \{\veta \in V'\ ;\ ||\veta - \mathring{\veta}|| < \epsilon\},\\
        L_h &= \{\veta \in V'\ ;\ \ell_0(\veta)  \leq m + h\}.
    \end{aligned}
    \]
    The set $L_h \cap B^c_{\epsilon}$ is the intersection of a closed set and a compact set, so it is compact. Therefore, the infimum of the continuous function $\ell_0(\veta)$ over $L_h \cap B^c_{\epsilon}$ is achieved at some $\veta'$. By the assumption of a unique minimum, $\ell_0(\veta') > m$ for all $\epsilon > 0$. Additionally, as $h > 0$, the infimum of $\ell_0$ over $L_{h}^{c}$ is strictly greater than $m$. Therefore, $\inf_{\veta \in B^c_{\epsilon}} \ell(\veta) > \ell(\mathring{\veta})$ and so $\mathring{\veta}$ is a well-separated minimum. }
    
    The deviation of $\ell_T$ from $\ell_0$ is controlled as,
    \begin{equation*}
    \begin{aligned}
    \sup_{\veta \in V'}  \lvert \ell_T(\veta) - \ell_0(\veta)  \rvert &= \sup_{\veta \in V'}  \lvert \tr \vec{R}_{\vec{x}\vec{x}}^{-1}(\veta) (\vec{S}_T-\bSigma_{\vec{x}\vec{x}}) \rvert \\
    &\leq \sup_{\veta \in V'} ||\vec{R}_{\vec{x}\vec{x}}^{-1}(\veta)||_F ||\vec{S}_{T}-\bSigma_{\vec{x}\vec{x}}||_F\\
    &\leq \epsilon^{-1}\sqrt{n}||\vec{S}_{T}-\bSigma_{\vec{x}\vec{x}}||_F.
    \end{aligned}
    \end{equation*}
    The third line follows from the definition of $V'$, which imposes that $\lambda_{\min}(\vec{R}_{\vec{x}\vec{x}}(\veta)) \geq \epsilon$. As the second moment of $\vec{x}$ exists, $\vec{S}_{T}$ is consistent so  $||\vec{S}_T-\bSigma_{\vec{x}\vec{x}}||_F \overset{p}{\to} 0$.  So, $\ell_T(\veta)$ converges uniformly in probability to the limiting function $\ell_0(\veta)$. As $\hat{\veta}_{T}$ minimizes $\ell_T$, standard results for $M$-estimators \cite[p. 45]{vaart1998} imply that $\hat{\veta}_{T} \overset{p}{\to} \mathring{\veta}$, so 
    $(\hat{\vec{A}}_T, \hat{\vec{B}}_T, \hat{\vec{\bPhi}}_T) \overset{p}{\to} (\mathring{\vec{A}}, \mathring{\vec{B}}, \mathring{\bPhi})$
    as well. 
    Note that if $\bSigma_{\vec{x}\vec{x}} \in \mathcal{R}(\vec{n}, \vec{r})$, then the minimizing $\mathring{\vec{R}}_{\vec{x}\vec{x}}$ is $\bSigma_{\vec{x}\vec{x}}$, which is the unique minimum by the Gibbs inequality. 
\end{proof}
\begin{proof}[Proof of Theorem \ref{thm:normality}]\label{sec:thm3proof} For any $r>0$, define the set of $\veta \in V'$ with $||\vec{R}_{\vec{x}\vec{x}}(\veta) - \bSigma_{\vec{x}\vec{x}}||_F < r$ and $||\vec{R}^{-1}_{\vec{x}\vec{x}}(\veta) - \bSigma_{\vec{x}\vec{x}}^{-1}||_F < r$. As $\vec{R}_{\vec{x}\vec{x}}(\veta), \vec{R}_{\vec{x}\vec{x}}^{-1}(\veta)$ are continuous and $\mathring{\veta}$ is in both sets, their intersection is a non-empty open neighborhood of $\mathring{\veta}$. In this neighborhood, the objective function $\ell_T(\veta)$ is Lipschitz, as for any  $\veta_1, \veta_2$  in the neighborhood with $\vec{R}_1 \equiv \vec{R}_{\vec{x}\vec{x}}(\veta_1), \vec{R}_2\equiv\vec{R}_{\vec{x}\vec{x}}(\veta_2)$, the difference $|\ell_{T}(\veta_1){-}\ell_{T}(\veta_2) |$ is bounded above,
\begin{equation*}
\begin{aligned}
|\ell_{T}(\veta_1){-}\ell_{T}(\veta_2) | &\leq \lvert \log \det \vec{R}_1\vec{R}_2^{-1} \rvert + \lvert \tr\vec{S}_T [\vec{R}_{1}^{-1} - \vec{R}_2^{-1}] \rvert \\
& \mkern-50mu\leq \vert\tr \vec{R}_1 [\vec{R}_2^{-1} \mkern-3mu- \vec{R}_1^{-1}]| + |\tr \vec{S}_T [\vec{R}_2^{-1}\mkern-3mu - \vec{R}_1^{-1}] |\\
& \mkern-50mu\leq \left(||\vec{R}_1 - \bSigma_{\vec{x}\vec{x}}||_F + ||\bSigma_{\vec{x}\vec{x}}||_F + ||\vec{S}_T||_{F}\right) \times\\
    &\qquad ||\vec{R}_1^{-1}||_{F}||\vec{R}_2^{-1}||_{F} ||\vec{R}_1 - \vec{R}_2||_F\\
&\mkern-50mu\leq m(r, \bSigma_{\vec{x}\vec{x}}, \vec{S}) ||\vec{R}_1 - \vec{R}_2||_F,
\end{aligned}
\end{equation*}
with $m(r, \bSigma_{\vec{x}\vec{x}}, \vec{S}) = (r+||\bSigma_{\vec{x}\vec{x}}||_F + ||\vec{S}||_F)(r+||\bSigma_{\vec{x}\vec{x}}^{-1}||_F)^2$.  Since by  assumption $E[||\vec{x}_1||^4] < \infty$, it is implied that $E[||\vec{S}_T||^2_{F}]<\infty$ and therefore $E[m(r, \bSigma_{\vec{x}\vec{x}}, \vec{S})^2] < \infty$. As $||\vec{R}||_F$ is bounded within the neighborhood, it is similarly true that the associated $||\vec{A}||_F, ||\vec{B}||_F$ are bounded. Therefore, 
\begin{equation*}
\begin{aligned}
||\vec{R}_1 - \vec{R}_2||_F &\leq 2 ||\vec{A}_1||_F ||\vec{A}_1 - \vec{A}_2||_F  + ||\bPhi_1 - \bPhi_2||_F \\&\quad + 2||\vec{B}_1||_F||\vec{B}_1 - \vec{B}_2||_F  \\
& \leq C(r, \bSigma_{\vec{x}\vec{x}}) ||\veta_1 - \veta_2||_2
\end{aligned}
\end{equation*}
for some non-random $C(r, \bSigma_{\vec{x}\vec{x}})$. Hence, $\ell_T(\veta)$ is Lipschitz within some neighborhood of $\mathring{\veta}$ with first differential
\begin{equation*}
\begin{aligned}
d\ell_{T}(\veta, d\veta) = &\tr(\vec{R}_{\vec{x}\vec{x}}^{-1}(\veta)d\vec{R}_{\vec{x}\vec{x}}(\veta, d\veta))\\&- \tr(\vec{R}_{\vec{x}\vec{x}}^{-1}(\veta)d\vec{R}_{\vec{x}\vec{x}}(\veta, d\veta)\vec{R}_{\vec{x}\vec{x}}^{-1}(\veta)\vec{S}_T).
\end{aligned}
\end{equation*}
Both $d\vec{R}_{\vec{x}\vec{x}}(\veta, d\veta)$ and $\vec{R}_{\vec{x}\vec{x}}^{-1}(\veta)$ exist for all $\veta \in V$, so $d\ell_{T}(\veta, d\veta)$ is well-defined.
The second differential expands as
\begin{equation*}
\begin{aligned}
d^2\ell_{T}(\veta, d\veta) &= 
\label{eq:2nddiff}
 2\tr\left([\vec{R}_{\vec{x}\vec{x}}^{-1}(\veta)d\vec{R}_{\vec{x}\vec{x}}(\veta, d\veta)]^2\vec{R}_{\vec{x}\vec{x}}^{-1}(\veta) \vec{S}_T\right)\\&-\tr\left([\vec{R}_{\vec{x}\vec{x}}^{-1}(\veta)d\vec{R}_{\vec{x}\vec{x}}(\veta, d\veta)]^2\right).
\end{aligned}
\end{equation*}
Taking expectation and evaluating at $\mathring{\veta}$, the above display equals $||\bSigma^{-1/2}_{\vec{x}\vec{x}}d\vec{R}_{\vec{x}\vec{x}}(\mathring{\veta}, d\veta)\bSigma^{-1/2}_{\vec{x}\vec{x}}||^2_F$ for $d\vec{R}_{\vec{x}\vec{x}}$ as in \eqref{eq:Rdifferential}, evaluated at $(\mathring{\vec{A}}, \mathring{\vec{B}}, \mathring{\bPhi})$ with tangent vector $(d\vec{A}, d\vec{B}, d\bPhi)$.  This norm is zero only if $d\vec{R}_{\vec{x}\vec{x}}(\mathring{\veta}, d\veta)$ is zero. However, as in the proof of Theorem \ref{thm:consistency}, identifiability of $\mathring{\veta}$ implies that $d\vec{R}_{\vec{x}\vec{x}}$ is non-zero for all non-zero $(d\vec{A}, d\vec{B}, d\bPhi)$. Let $\vec{V}_{0}$ be the Hessian matrix of $\ell_0$ at $\mathring{\veta}$. As the second differential is positive, 
the quadratic form $d\veta^{\tran} \vec{V}_{0} d\veta$ is positive for all $d\veta \neq \vec{0}$, so $\vec{V}_0$ is positive definite. Standard results for M-estimators (e.g. \cite[p. 53]{vaart1998}) then imply that $\sqrt{T}(\hat{\veta}_T - \mathring{\veta}) \overset{d}{\to} \mathcal{N}\left(\vec{0}, \vec{W} \right)$, 
where
\begin{equation}
\label{eq:asymptoticcov}
\vec{W} = \vec{V}_0^{-1}E\Big[\frac{\partial \ell_T(\mathring{\veta})}{\partial \veta} \frac{\partial \ell_T(\mathring{\veta})^{\tran}}{\partial \veta}  \Big] \vec{V}_0^{-1}.
\end{equation}
\end{proof}

\end{appendices}

\end{document}